\documentclass[letterpaper,11pt]{article}
\usepackage{hyperref}

\usepackage{amsmath}
\usepackage{amsfonts}
\usepackage{amsthm}
\usepackage{amssymb}
\usepackage{bbm}
\usepackage[only,llbracket,rrbracket]{stmaryrd}

\usepackage{mathtools}
\usepackage{tikz}
\usetikzlibrary{graphs}
\usetikzlibrary{trees,arrows}
\usetikzlibrary{graphs.standard}
\usetikzlibrary{shapes.geometric}
\usetikzlibrary{arrows.meta}
\tikzset{
  n/.style={shape=circle, draw=black, minimum size=20pt, font=\small},
  diredge/.style={->, >=stealth, thick},
  undiredge/.style={-, >=stealth, thick},
}

\usepackage{float}
\usepackage{subcaption}

\usepackage{algorithm}
\usepackage{algpseudocode}

\usepackage[style=alphabetic,natbib=true,maxcitenames=3]{biblatex}
\usepackage[margin=1in]{geometry}

\newtheorem{thm}{Theorem}[section]
\newtheorem{cor}[thm]{Corollary}
\newtheorem{prop}[thm]{Proposition}
\newtheorem{lmm}[thm]{Lemma}
\theoremstyle{definition}
\newtheorem{defi}[thm]{Definition}
\newtheorem{rmk}{Remark}
\newtheorem{conj}{Conjecture}

\newcommand{\N}{\mathbb{N}}

\newcommand{\E}{\mathbb{E}}
\newcommand{\card}[1]{\left| #1 \right|}
\newcommand{\ind}[1]{\mathbbm{1}_{#1}}
\DeclareMathOperator{\ballop}{\mathcal{B}}
\newcommand{\ball}[2]{\ballop_{#1} \left(#2\right)}

\DeclarePairedDelimiter{\ios}{\{}{\}}

\newcommand{\s}{\ios*}

\newcommand{\nashe}{\textsc{Nash} equilibrium}
\newcommand{\nashes}{\textsc{Nash} equilibria}

\newcommand{\unG}{\tilde{G}}
\newcommand{\revE}{\hat{E}}

\title{Edge-dominance games on graphs}
\author{Farid Arthaud\\
{\footnotesize MIT CSAIL}\\
{\footnotesize Cambridge, Massachusetts}\\
{\footnotesize \href{mailto:farto@csail.mit.edu}{\texttt{farto@csail.mit.edu}}}
\and
Edan Orzech\\
{\footnotesize MIT CSAIL}\\
{\footnotesize Cambridge, Massachusetts}\\
{\footnotesize \href{mailto:iorzech@csail.mit.edu}{\texttt{iorzech@csail.mit.edu}}}
\and
Martin Rinard\\
{\footnotesize MIT CSAIL}\\
{\footnotesize Cambridge, Massachusetts}\\
{\footnotesize \href{mailto:rinard@csail.mit.edu}{\texttt{rinard@csail.mit.edu}}}}
\date{}

\begin{document}
\maketitle
\thispagestyle{empty}

\begin{abstract}
	We consider zero-sum games in which players move between adjacent states,
	where in each pair of adjacent states one state dominates the other.
	The states in our game can represent positional advantages in physical
	conflict such as high ground or camouflage, or product characteristics
	that lend an advantage over competing sellers in a duopoly.
	We study the equilibria of the game as a function of the topological and
	geometric properties of the underlying graph.
	Our main result characterizes the expected payoff of both players
	starting from any initial position, under the assumption that the graph
	does not contain certain types of small cycles.
	This characterization leverages the block-cut tree of the graph, a
	construction that describes the topology of the biconnected components of
	the graph.
	We identify three natural types of (on-path) pure equilibria, and
	characterize when these equilibria exist under the above assumptions.
	On the geometric side, we show that strongly connected outerplanar graphs
	with undirected girth at least \( 4 \) always support some of these
	types of on-path pure equilibria.
	Finally, we show that a data structure describing all pure equilibria can
	be efficiently computed for these games.
\end{abstract}
\newpage

\clearpage
\setcounter{page}{1}

\section{Introduction}
Consider two players moving among the vertices of a graph~\( G \), where at each
timestep both players simultaneously move to a neighbor of their current vertex,
or remain.
The game ends with some constant probability~\( (1-\delta) \in \left] 0; 1
\right[ \) at each timestep.
The edges of the graph are undirected for movement but directed for payoff:
at the conclusion of the game, each player receives a payoff of \( 1 \) if they
are at a parent of the other player, \( -1 \) if the other player is at one of
their parents, and \( 0 \) if their vertices are not neighbors (or if they are at
the same vertex).

This setting is a general way of constraining how players change actions between
rounds in a repeated symmetric zero-sum game: each vertex of the graph \( G \) is
a possible action, the edges of the graph describe the payoffs of the game (the
symmetric zero-sum game matrix is seen as an adjacency matrix) and players can
only move between \emph{similar} actions between rounds (i.e.\ actions that have
an edge between them).

These games model aspects of physical conflict in which adjacent vertices
represent positions where adversaries can engage.
For instance, terrain features such as high ground or easy camouflage may yield
an advantage to one of the adjacent positions.
These games also model aspects of duopoly markets in which each vertex of the
graph represents a set of features of a product manufactured by two competing
companies.
The companies can only incrementally change their product, by moving along the
edges of the graph.
Adjacent sets of features compete for the same consumers, with the direction of
the edge indicating which product is more profitable to sell when these two
products are on the market (as a function of consumer preference, production
prices and selling prices).
On the other hand, products further away from each other in the graph have
features that appeal to sufficiently different segments of consumers that they
have no discernible advantage over one another.
In turn, this can be seen as a dynamic variant of discrete \textsc{Hotelling}
models~\citep{hotelling} (recall location in \textsc{Hotelling} models is often
interpreted as product differentiation), an area of interest in economics and
algorithmic game
theory~\citep{RePEc:upf:upfgen:96,DBLP:conf/esa/DurrT07,DBLP:journals/ijgt/Fournier19}.

\bigskip{}
We focus on characterizing the payoff that each player can obtain from any given
initial position in a \nashe{} of the game.
We study the conditions under which one player has a \emph{winning strategy},
namely a strategy that gives the player strictly positive expected payoff, as
well as how much payoff this player can obtain in expectation.
We also study the conditions under which both players have \emph{safe
strategies}, namely strategies where they do not incur strictly negative payoff.

We identify three types of on-path pure \nashes{} (meaning players do not mix
unless a player has deviated) in safe strategies that are extremal equilibria in
cycle graphs~\( C_n \) (for \( n \geq 4 \)).
The first type is the \emph{$k$-chase} equilibrium, in which one player
follows \( k \)~steps behind the other player, without ever reaching them.
The second type is the \emph{walking together} equilibrium, in which both players
start at the same vertex and always deterministically move together to a same
vertex at each round, ensuring a stalemate.
The third type is the \emph{static} equilibrium, where both players remain static
at a distance from each other.
We provide characterizations under which $k$-chase, walking together and static
equilibria exist on certain classes of graphs.
These characterizations imply sufficient conditions for the existence of on-path
pure \nashes{} in safe strategies in general.

We also consider the existence of these types of equilibria in \emph{planar
graphs}.
These graphs are important as possible models of maps of geographical locations,
and because as our results will imply in Sections~\ref{sec:constructions}
and~\ref{sec:outer}, non-planarity plays a role in determining the existence of
the above equilibria.
We provide sufficient conditions on \emph{outerplanar graphs}, a particular class
of planar graphs, such that walking together and \( 2 \)-chase equilibria exist.

\subsection{Results and techniques}\label{subsec:results}
\paragraph{Main result}
Our main result characterizes the value of the game for any given initial
position, under the assumption that the graph does not contain certain types of
small cycles as a subgraph.
\begin{thm}[Informal version of Theorem~\ref{thm:char4}]\label{thm:char4inf}
	In any graph of undirected girth at least \( 4 \) that does not contain
	\emph{unbalanced} small cycles, a player is in a losing position if and
	only if their resulting connected component when cutting the graph
	halfway between the two players is always an (out)-directed rooted tree,
	rooted at the vertex where the graph was cut, and the two players are
	separated by this cut.
\end{thm}
More specifically, the cycles that we exclude are all \( 3 \)-cycles and
the four \emph{unbalanced} cycles of length~\( 4 \) or \( 5 \), which are cycles
with edges oriented such that there is exactly one maximal vertex (with no incoming
edges) and one minimal vertex (with no outgoing edges).
In particular, our characterization holds for trees and more generally all graphs
of (undirected) girth at least \( 6 \).

The key concept in this characterization's precise formulation is the position
of the players in the \emph{block-cut tree} of the (undirected) graph.
The block-cut tree of a graph is the tree whose vertices are the maximal biconnected
components and the cut vertices of the graph, with edges between each biconnected
component and all of its cut vertices.
The block-cut tree is important because we establish that players have a safe
strategy as long as they find themselves in a biconnected component.
Intuitively, if cutting the graph halfway between the players leaves a player
with biconnected components (or other types of safe positions) in their half,
they can reach this safe position before the other player reaches them.

\paragraph{Static and cycle-based equilibria}
We then characterize the existence of walking together, \( 2 \)-chase, and static
equilibria.
Under the same topological assumptions as above we show that walking together and
\( 2 \)-chase equilibria exist if and only if the graph contains a directed cycle
-- the weakest possible condition we could hope for.
We also provide an exact characterization of graphs with static equilibria under
these assumptions (Proposition~\ref{prop:nostatic}).
We then show the importance of small unbalanced cycles for these
characterizations by presenting constructions that contain some small unbalanced
cycles, verify our characterization above, and yet do not admit these equilibria
-- despite being strongly connected and of girth \( 5 \) (see
Section~\ref{sec:constructions} for the constructions).

We provide another sufficient condition using graph geometry rather than
topology, showing that any strongly connected \emph{outerplanar} graph with
undirected girth at least \( 4 \) supports a \( 2 \)-chase and a walking together
equilibrium (Theorem~\ref{thm:outerplanar}).
An outerplanar graph is a planar graph in which all of the vertices belong to the
outer face of the graph (for some planar embedding of the graph). Intuitively,
this result holds for two reasons. One, such a graph has a face whose edges form
a directed cycle, which gives a player some safety from the other player's
deviations when they remain on it.
Two, because of outerplanarity, a player who exits that directed cycle from some
vertex $v$ can only reenter it at a vertex that is either $v$ or its neighbors on
the cycle. Therefore, the player who does not deviate can remain at a safe
distance from the deviator at all times. This result falls into the larger
context of the existence of static and cycle-based equilibria in \emph{planar
graphs}.
We conjecture in Section~\ref{sec:outer} that all strongly connected planar
graphs with girth at least \( 4 \) have a static or cycle-based equilibrium,
given our result on outerplanar graphs and that our constructions without
cycle-based equilibria are non-planar.

\paragraph{Computational aspects}
We show that a data structure describing all pure equilibria can be computed
efficiently, despite the number of pure equilibria potentially being exponential.
Computing all mixed equilibria in the game (for all starting positions) results
from a straightforward application of \textsc{Bellman}'s equations and value
iteration algorithms.

\paragraph{Roadmap}
The remainder of this section contains related work.
In Section~\ref{sec:prelim} we define our game formally, and characterize
equilibria in cycle graphs to introduce important definitions and lemmas.
Section~\ref{sec:trees} characterizes winning positions in trees and
Section~\ref{sec:g6} characterizes all equilibria in graphs of girth at least \(
6 \).
The results from these two sections are combined to obtain our main results in
Section~\ref{sec:g4} for graphs with no small unbalanced cycles.
Section~\ref{sec:constructions} contains our strongly connected constructions
without cycle-based or static equilibria.
Section~\ref{sec:outer} contains our results for outerplanar graphs, and finally
Section~\ref{sec:comp} presents our algorithm to efficiently compute pure
equilibria.

\subsection{Related work}
\paragraph{Pursuit-evasion games}
The branch of research conceptually closest to ours is pursuit-evasion games on
graphs, also known as cops and robbers problems, originally introduced
by~\citet{quilliot1978these}.
In this game, a robber player chooses a vertex in a finite graph, after which \( k
\) cops choose positions in the graph.
Then, turn by turn, the robber and cops move along the edges in the graph, until
either a cop reaches the same vertex as the robber or until a configuration repeats
twice.
In the latter case, the robber manages to evade the cops and wins, whereas if a
cop reaches the robber, the robber loses.
In the original paper, \citet{quilliot1978these} characterizes the graphs in
which one cop is sufficient to capture the robber, and a long line of work
followed on finding bounds on the \emph{cop number} of a graph (the number of
cops necessary to capture the
robber)~\citep{bonato2011game,DBLP:journals/tcs/KonstantinidisK16,DBLP:journals/tcs/BonatoCP10}.

\citet{bonato2011game} provide a general survey of variants of cops and robbers.
Several variants have some similarities to our setting.
\citet{DBLP:journals/tcs/KonstantinidisK16} study simultaneous-move cops and
robbers and show that the (appropriately defined) cop number of a graph is
unchanged relative to the classical cops and robbers.
\citet{DBLP:journals/ejc/Hamidoune87} introduces a variant on directed graphs,
however the direction of edges has a different meaning than in our game: players
are constrained to follow edges only in one direction, and the goal of the cops
is still to reach the same vertex as the robber.
\citet{DBLP:journals/tcs/BonatoCP10} introduce capture from a distance, where
cops win if any one of them comes within a certain distance of the robber.
As cops and robbers has traditionally been studied from a graph-theoretical or
combinatorial point of view, some works investigate game-theoretic formulations,
such as the recent work of \citet{DBLP:journals/dga/KehagiasK21}, and the survey
by \citet{DBLP:journals/dga/Luckraz19} of game-theoretic formulations.

There are several important differences between the game of cops and robbers and
our model:
	(i)~in our model the players are symmetric whereas in cops and robbers
	there is a pursuer and an evader;
	(ii)~the winning condition in our model is to reach a parent vertex of the
	other player, whereas in standard cops and robbers it is to reach the
	same vertex (and edges are undirected);
	(iii)~our model is simultaneous move whereas standard cops and robbers
	take turns.
These points mean that results from cops and robbers do not apply in our setting.
In particular, players can find themselves mixing between situations where they
have strictly positive, zero and strictly negative payoff (whereas in usual cops
and robbers there are no draws -- one player can force a win by
\textsc{Zermelo}'s theorem -- see Remark~\ref{rmk:3cyckk} for a simple example).
Moreover, no player can win with probability~\( 1 \) in our
model and players randomize to evade capture (see Remark~\ref{rmk:pathkk} for a
simple example, and Proposition~\ref{prop:algpure} for a proof).
In contrast, even in simultaneous cops and robbers, the cop number is defined as
the minimum number of cops such that the robber is captured with probability~\( 1
\)~\citep{DBLP:journals/tcs/KonstantinidisK16}, significantly altering the
analysis.
The main idea in simultaneous cops and robbers is for the cop to guess the next
move of the robber, and play as if their guess is correct -- with probability~\(
1 \), they will eventually guess correctly for sufficiently long that they will
capture the robber.
Such a strategy is not viable in our game since it requires that for every pair
of positions, the cop can win against the robber in the usual turn-based cops and
robbers.
Since our players are symmetric, if one player has a superior position to the
other then the converse cannot be true, and a guessing strategy can lead the
player in a superior position to end up in an inferior position by misguessing.
The optimal strategies in our setting can therefore ensure positive expected
payoff at best, but never capture with probability~\( 1 \) as in cops and robbers
and its variants -- and they very well can lead a player with strictly positive
expected payoff to obtain strictly negative payoff with non-zero probability.

More importantly, beyond our characterization of winning positions in the game,
many of our results concern properties of \( 0 \)-payoff equilibria the players
can be in, characterizing player behavior when neither player has an advantage
over the other.
To the best of our knowledge this has not been explored for cops and robbers,
where most work only focuses on characterizing when one player has an advantage
over the other (in particular through the cop
number)~\citep{DBLP:journals/dga/Luckraz19}.

\paragraph{Stochastic games}
Our game is an instance of a stochastic game~\cite{shapley1953stochastic}, i.e.\ an extensive-form game with a
state that is affected by the actions of both players (and potentially also
external randomness, but not in our case) and which affects the players' payoffs.
However, it has much more added structure which makes our analysis possible: the
state space is a cartesian product (of the graph's vertices with itself), each
player affects only one component of the state, and the payoff is related to the
allowed transitions (through the graph edges).
As far as we are aware, there are no results in the area of stochastic games
concerning games with this structure.

\paragraph{Discrete \textsc{Hotelling} models}
Finally, as mentioned earlier, our model has similarities with discrete
\textsc{Hotelling} models, such as those presented in
\citet{RePEc:upf:upfgen:96}.
These games are sometimes also called \textsc{Voronoi} games on graphs or
\emph{competitive facility location} games when the players can only locate at
vertices (similarly to our model).
In these models, two players choose a vertex of a graph (or sometimes a
position along an edge of a graph) and then are rewarded as a function of the
quantity of vertices or edges (sometimes weighted) that are closer to them than
to the other player~\citep{DBLP:journals/ijgt/Fournier19}.
The fundamental difference between these models and ours is that only adjacent
vertices in the graph have an advantage over one another in our model, whereas in
\textsc{Voronoi} games it is likely that most pairs of distinct vertices have
unequal payoff.
Moreover, to the best of our knowledge, these models are static and do not model
dynamics of relocation like ours.
Some works analyze best response dynamics in these
games~\citep{DBLP:conf/esa/DurrT07}, but with no restrictions on where players
can move to.

\section{Preliminaries}\label{sec:prelim}
Let \( G = (V,E) \) be a connected oriented graph, and \( \delta \in ] 0; 1 [ \).
An oriented graph is a directed graph with no loops or parallel edges, they are
the graphs obtained by assigning an orientation to each edge in an undirected
graph.
If \( (u, v) \in E \), we often write \( (u \to v) \in E \) or simply \( u \to v
\) when \( E \) is clear from context.
We write $u-v$ to say that there is an edge \( u \to v \) or \( v \to u \) in the
graph.
For convenience, we denote \( \revE = \{ (v,u) \mid (u,v) \in E \} \) the
reversed set of edges.
\begin{defi}
	A \textbf{path} in \( G \) is a list of at least \( 2 \) distinct
	vertices \( u_0 \to u_1 \to \cdots u_k \) with directed edges, whereas an
	\textbf{undirected path} is a list of at least \( 2 \) distinct vertices
	\( u_0, u_1, \ldots, u_k \) such that for each \( i \in \llbracket 0, k-1
	\rrbracket \), \( u_i \to u_{i+1} \) or \( u_i \leftarrow u_{i+1} \).
	We denote \( \unG \) the \emph{undirected graph} underlying \( G \).
	An \textbf{undirected cycle} in \( G \) is a cycle in \( \unG \), i.e.\
	an undirected path of length at least \( 3 \) such that the last vertex is
	a neighbor of the first.
	The \textbf{ball} of radius~\( r \) centered at vertex~\( u \) is,
	\[
		\ball{r}{u} = \{ v \in V \mid
			\exists k \in \llbracket 0; r \rrbracket,\;
			\exists v_0, \ldots, v_k, \;
			(u = v_0) - v_1 - \cdots - v_{k-1} - (v_k = v)
			\}.
	\]
	\( g(\unG) \) denotes the \textbf{undirected girth} of \( G \), i.e.\ the
	length of a shortest cycle in \( \unG \).
	We refer to it as \( g \) when \( G \) is clear from context.
\end{defi}

The game is defined by the graph \( G \) and initial positions for both players
\( (x_0, y_0) \in V^2 \).
At each timestep \( t \in \N \), we denote \( (x_t, y_t) \) the positions of the
players in the graph.
The strategies of the players are mappings from their current positions \( (x_t,
y_t) \) to a distribution over their neighborhoods \( \ball{1}{x_t} \) and \(
\ball{1}{y_t} \).
We denote \( \varphi_x: V^2 \to \Delta(V) \) the strategy of player~\( x \) and
\( \varphi_y \) for player~\( y \), where \( \Delta(V) \) is the simplex over the
vertices of \( G \).
For a given initial state~\( s_0 \), a distribution over histories of play \(
{\left( h_t \right)}_{t \in \N} \) is naturally induced by \( \varphi_x \) and \(
\varphi_y \): we write \( h \sim (\varphi_x, \varphi_y) \) for a random variable
\( h \) following this distribution when \( (x_0, y_0) \) is clear from context.
Note the game is defined in such a way that the strategies are memoryless: they
only depend on the current state and not on the history of play.

The game ends with probability \( (1-\delta) \) at the end of each round, for
some fixed parameter \( \delta \in \left] 0; 1 \right[ \).
At that point, the payoff of each player is \( 1 \) if they are at a parent of
the other player, \( -1 \) if they are at a child of the other player, and \( 0
\) otherwise (in particular, if both players are at the same vertex).
The game is a zero-sum game, and the expected payoff of player~\( x \) can be
written,
\begin{equation}\label{eq:defipayoff}
	u_x(\varphi_x, \varphi_y) = (1-\delta) \, \E_{h \sim (\varphi_x,
	\varphi_y)} \left[ \sum_{t \in \N} \delta^t \left( \ind{h_t \in E}
	- \ind{h_t \in \revE} \right) \right].
\end{equation}
As is often done in the repeated games literature and in order to simplify
analysis, we often interpret equation~\eqref{eq:defipayoff} as the payoff of a
discounted game: the game is then always infinite, and payoff at round \( t \) is
multiplied by a factor \( \delta^t \).
We also refer to the sum starting at \( t = 1 \) in
equation~\eqref{eq:defipayoff} as the \textbf{continuation payoff} of player~\( x
\).

\begin{defi}
	A pair of strategies is called a \nashe{} if neither player can increase
	their expected payoff by changing strategies.
	Note that by the minmax theorem, there always exists a \nashe{}.
	We call the \textbf{value} of a vertex~\( u \) over a vertex~\( v \) the
	minmax equilibrium value of player~\( x \) when players start at \( (u,
	v) \).
	A strategy is called \textbf{safe} for a player if its expected payoff is
	(weakly) positive, and \textbf{winning} if its expected payoff is
	strictly positive.
\end{defi}

\subsection{Cycle-based and static equilibria}\label{subsec:types}
We begin by considering cycle graphs, in which we characterize all equilibria and
identify three particular types of extremal pure equilibria.
This leads us to define the three types of equilibria we investigate in general
graphs in the following sections.
We also introduce important definitions and lemmas for the rest of the paper by
studying the directed \( 3 \)-path.

\begin{figure}
	\begin{subfigure}{.3\textwidth}
		\centering
		\begin{tikzpicture}
			\graph[clockwise, empty nodes, nodes={circle,fill=black},
			edges={>=stealth}] {subgraph C_n [n=3, ->]};
		\end{tikzpicture}
		\caption{The directed \( 3 \)-cycle.}\label{fig:3cyc}
	\end{subfigure}\hfill
	\begin{subfigure}{.3\textwidth}
		\centering
		\begin{tikzpicture}
			\graph[clockwise, empty nodes, nodes={circle,fill=black}, edges={>=stealth}] {subgraph C_n [n=4, ->]};
		\end{tikzpicture}
		\caption{The directed \( 4 \)-cycle.}\label{fig:4cyc}
	\end{subfigure}\hfill
	\begin{subfigure}{.3\textwidth}
		\centering
		\begin{tikzpicture}
			\graph[grow down, nodes={circle, draw=black}, edges={>=stealth}] {T -> M -> B};
		\end{tikzpicture}
		\caption{The directed \( 3 \)-path.}\label{fig:3path}
	\end{subfigure}
	\caption{Example graphs.}
\end{figure}
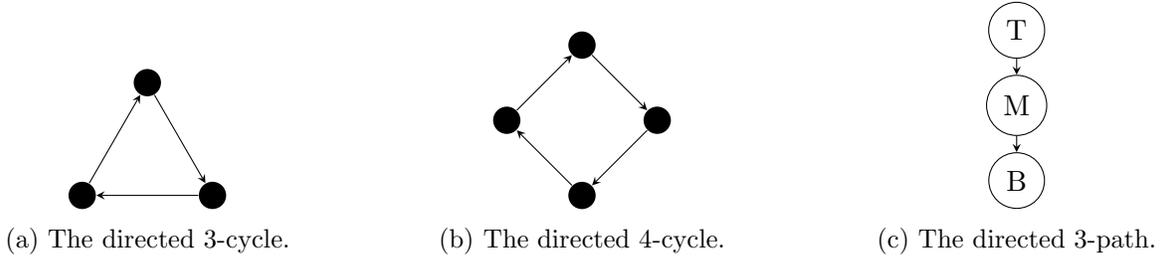

\paragraph{\( 3 \)-cycle}
In Figure~\ref{fig:3cyc} the directed \( 3 \)-cycle is shown: on this graph, our
game reduces to (repeated) rock-paper-scissors game, as the players are
unconstrained in which actions they can choose at each round.
Therefore, there exists a unique \nashe{} where both players uniformly mix over
all three possible movements (stay or move to one of their two neighbors).

\begin{rmk}\label{rmk:3cyckk} Note this example highlights a fundamental
	difference with cops and robber games: both players are mixing between
	outcomes that have strictly
	positive, zero, or strictly negative value for them, meaning players need
	to take a chance of capturing and a risk of being captured.
	This cannot occur in cops and robbers for several reasons: the moves are
	not simultaneous and the roles are asymmetric (the cop cannot be
	captured).
	In simultaneous cops and
	robbers~\citep{DBLP:journals/tcs/KonstantinidisK16}, it would be a
	cop-win graph since the cop will eventually collide with the robber with
	probability~\( 1 \) regardless of starting positions.
\end{rmk}

\paragraph{\( 4 \)-cycle}
The directed \( 4 \)-cycle in Figure~\ref{fig:4cyc} has more equilibria.
If both players start opposite from each other, their optimal strategy is to
randomize between their two neighbors with any distribution that puts at most
\( 1/2 \) probability on their parent.
Indeed, if one player moves counterclockwise to their parent with probability
strictly more than \( 1/2 \), the other player can ensure strictly positive
payoff by remaining at their current vertex; if one player puts any probability on
staying at their vertex then the other player can ensure strictly positive payoff
by moving to their child (clockwise).

If both players start at the same vertex, the optimal strategies are similar: any
mixing between their two neighbors that puts at most \( 1/2 \) probability on
their child is optimal.
In both starting positions the minmax one-round expected payoff is \( 0 \) for
both players, and by induction the overall minmax expected payoff is also \( 0
\).
The set of equilibria consists of all distributions that satisfy the above
conditions, therefore the players will find themselves either at the same vertex
or at opposite vertices at every round if they start in one of these positions.
Note two particular extremal equilibria in this graph are: (i) both players are
at the same vertex and move deterministically counterclockwise to the parent
(together) and (ii) both players are opposite from each other and move
deterministically clockwise to their child at each round.

In a longer cycle, more strategies exist: when players are far from each other,
all actions are equivalent, whereas when the players are at distance~\( 2 \) or
\( 3 \) from each other, one player will have to avoid the other (by moving in
the opposite direction).
In particular, another type of on-path pure equilibrium arises when the cycle is
of length at least \( 6 \): both players can remain static at vertices that are
distance at least \( 3 \) from one another.
The following definition generalizes the three extremal equilibria we have seen
in cycles so far.
\begin{defi}
	A \textbf{walking together equilibrium} (WT) is an equilibrium such that
	for every \( t \in \N \), \( x_t=y_t \) and \( x_{t+1} \ne x_t \) with probability~\( 1 \).
	A \textbf{\( k \)-chase equilibrium}, for \( k \in \llbracket 2; +\infty
	\llbracket \), verifies \( x_{t+k} = y_t \) for all \( t \in \N \).
	A \textbf{static equilibrium} is such that with probability~\( 1 \),
	there is a $t_0$ and vertices \( x_\infty, y_\infty \in V \) such that
	for every $t\ge t_0$, \( (x_t, y_t) = (x_\infty, y_\infty) \).
\end{defi}

Before our final example, we define an important property of certain graphs that
simplifies analysis of equilibria.
In all generality, having negative payoff at a given round could still lead to
compensations later on, for instance a player could accept one immediate round of
negative payoff to ensure many rounds of positive payoff later on.
We define edges and graphs for which this does not have an effect on winning
strategies.
\begin{defi}
	For a given value of \( \delta \), an edge \( u \to v \) is called
	\textbf{decisive} if the value of a player at \( u \) over a player at \(
	v \) is strictly positive.
	A graph \( G \) is called \textbf{edge-decisive} for a given \( \delta \)
	if all of its edges are decisive.
	\( \delta \) will be omitted when clear from context.
\end{defi}

\begin{rmk}\label{rmk:wtcyc}
	In an edge-decisive graph, a walking together equilibrium always steps
	towards parents: \( \forall t,\; (x_{t+1} \to x_t) \in E \).
	This is because when walking together to a child, either player could
	make a profitable deviation by not moving.
	Inversely, a \( 2 \)-chase equilibrium in an edge-decisive graph always
	steps towards children: \( \forall t,\; (y_t \to y_{t+1}) \in E \)
	(otherwise the chased player can stop after having taken an edge in the
	opposite direction).
	In particular, both types of equilibria correspond to a directed cycle in
	the graph~\( G \).
\end{rmk}

\paragraph{\( 3 \)-path}
To illustrate cases where one player has an advantage over the other, we look at
the directed \( 3 \)-path example illustrated in Figure~\ref{fig:3path}, when one
player lies at \( T \) and the other at \( B \).
The one-step game is equivalent to matching pennies (where the two sides of the
pennies are `move' or `stay'): the top player wins if exactly one of them moves
to \( M \) whereas the bottom player wins if either both or neither of them move.
However, the two outcomes where the bottom player wins the one-shot game are not
equivalent: if they both move, the continuation payoff is \( 0 \), since both
players will simply move to \( T \) in the next round.
If neither moves, the game repeats and the top player has some chance of winning
again.
Similarly, if just the top player moves then they gain payoff~\( 1 \) and the
game repeats (since the players have the same two actions each, the top player
going to \( B \) is dominated).
If just the bottom player moves, the game ends at the next round with both
players reaching \( T \).
Regardless, the top player has strictly positive expected payoff starting from
the initial condition.
\begin{lmm}\label{lmm:3path}
	In the \( 3 \)-path illustrated in Figure~\ref{fig:3path}, a player at \(
	T \) has strictly positive payoff over a player at \( B \).
\end{lmm}
\begin{proof}
	Fix the strategy of the top player to uniformly randomize between staying
	and moving when in these positions, and to move to (or stay at) \( T \)
	in all other positions.
	Fix the other player's strategy arbitrarily, it is sufficient to show
	this leads to strictly positive payoff for the top player.
	We can assume without loss of generality that the game ends with
	payoff~\( 0 \) when both players reach the same vertex.
	Therefore, every round has (weakly) positive payoff for the top player,
	and the first round has strictly positive payoff for them (since the
	other player cannot match their move with probability~\( 1 \)).
\end{proof}

\begin{rmk}\label{rmk:pathkk}
	Note that unlike cops and robbers, the bottom player always has some
	probability of avoiding capture, by randomizing between moving to the
	middle vertex and staying.
	This is true of any position in the graph: if a player randomizes between
	all of their parents and staying at their current vertex, the other player
	cannot ensure capture in one round (and ensuring capture in one round has
	higher payoff than any strategy that does not, hence even in minmax play
	capture is never ensured).
\end{rmk}

\section{Trees}\label{sec:trees}
A simple generalization of the ideas behind the case of a single path is a tree,
where we find a characterization of positions that have positive minmax value for
a player.
We first define some useful notions to express our results.
\begin{defi}
	For a tree $T$ rooted at $r$ and a vertex $v$, let $T_v$ be the subtree
	of $T$ rooted at $v$.
	In a rooted tree, a directed edge \( u \to v \) is said to be pointing
	\textbf{upwards} if \( u \neq r \) and \( v \) is on the (undirected)
	path from \( u \) to \( r \) -- otherwise, it is a \textbf{downwards}
	edge.
	A rooted directed tree is called \textbf{outgoing} if all of its edges
	are downwards edges.
\end{defi}

The main intuition behind the characterization is to root the tree at the
midpoint of the path between the two players' positions.
If either player can reach an upwards edge in the tree without going through the
root, they are safe -- it ensures the other player would have to go
through their child to reach them.
Otherwise, the other player can reach the root and then start chasing them down
(all edges go downwards) until they reach a leaf, and Lemma~\ref{lmm:3path} for
Figure~\ref{fig:3path} shows they can obtain strictly positive payoff.
Figure~\ref{fig:treewin} shows the two situations where player~\( y \) has a
winning strategy over player~\( x \).

\begin{figure}
	\begin{subfigure}{.49\textwidth}
		\centering
		\includegraphics{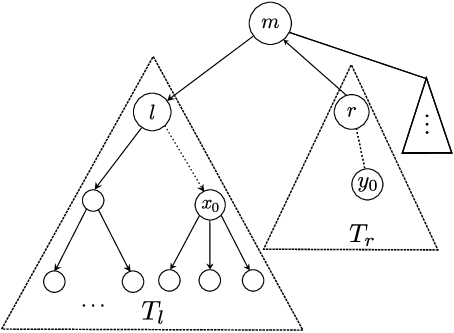}
	\end{subfigure}
	\begin{subfigure}{.49\textwidth}
		\centering
		\includegraphics{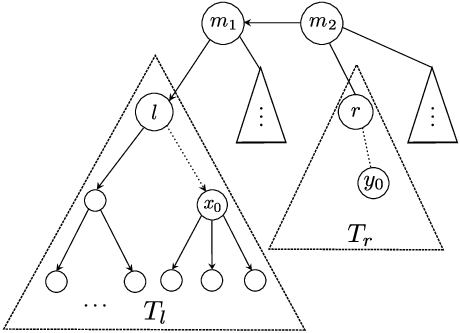}
	\end{subfigure}
	\caption{The two winning configurations for player~\( y \) in a tree.}\label{fig:treewin}
\end{figure}

\begin{thm}\label{thm:tree}
	For a given initial position $(x_0,y_0)$ in a tree, player~\( y \) has strictly
	positive payoff over player~\( x \) if and only if,
	\begin{enumerate}
		\item The path between $x$ and $y$ is of even length, has the
		form $x_0- \cdots -l \leftarrow m \leftarrow r- \cdots -y_0$, where
		$d(x_0, m)=d(y_0, m)$ and the subtree $T_l$ (containing player
		$x$) of $T$ rooted at \( m \) has no upwards edges (left tree in
		Figure~\ref{fig:treewin}); or
		\item The path between $x$ and $y$ is of odd length, has the form
		$x_0- \cdots -l \leftarrow m_1\leftarrow m_2-r- \cdots -y_0$, where
		$d(x_0, m_1)=d(y_0, m_2)$ and the subtree $T_l$ (containing player
		$x$) of $T$ rooted at \( m_2 \) has no upwards edges (right tree
		in Figure~\ref{fig:treewin}).
	\end{enumerate}
\end{thm}
\begin{proof}
	Suppose the path between \( x_0 \) and \( y_0 \) has even length, and is
	of the form\\
	$x_0- \cdots -l - m - r- \cdots -y_0$, where $d(x_0, m)=d(y_0, m)$.
	We root the tree at \( m \).
	We first show that if the edge between \( l \) and \( m \) is directed
	towards \( m \) (meaning \( (l \to m) \in E \)), then player~\( x \) has
	(weakly) positive payoff.
	Indeed, player~\( x \) can reach \( l \) while staying distance at least
	\( 2 \) from player~\( y \) by definition of \( m \).
	Then, by remaining at \( l \), player~\( y \) will be forced to go
	through \( m \) before reaching player~\( x \) or any of its ancestors.
	Moreover, if player~\( y \) reaches \( l \) then both players achieve
	payoff~\( 0 \) by playing a minmax strategy (by symmetry).

	Similarly, if the edge between \( m \) and \( r \) is directed as \( m
	\to r \), then player~\( x \) can safely reach \( l \), then \( m \).
	By then, player~\( y \) is either at \( m \) or within \( T_r \) (since
	they have traveled distance at most \( d(x_0,m) = d(y_0,m) \)).
	If both players are at \( m \) their payoff is \( 0 \) by symmetry (by
	example, both moving to a parent at each step until they reach a maximal
	vertex is an equilibrium).
	Otherwise, by the same reasoning as above, player~\( x \) is safe by
	remaining at \( m \).

	We now assume that the path has form $x_0- \cdots -l \leftarrow m
	\leftarrow r- \cdots -y_0$.
	Suppose there is an upwards edge \( a \to b \) in \( T_l \), then by the
	same reasoning player~\( x \) can reach \( l \) while staying distance at
	least \( 2 \) from player~\( y \).
	From there, player~\( x \) can go down to \( a \), maintaining
	distance at least \( 2 \) from player~\( y \).
	Then, the same reasoning applies to show that remaining at vertex~\( a \)
	is safe for player~\( x \).

	Conversely, if there are no upwards edges in \( T_l \), player~\( y \)
	has a winning strategy by first reaching \( r \), then mixing between
	moving towards player~\( x \) and remaining at their current position.
	By reasoning by induction on the length of the longest path below
	player~\( x \) and as seen in Lemma~\ref{lmm:3path}, this yields strictly
	positive payoff for player~\( y \).

	\bigskip{}
	We now assume that the path between \( x_0 \) and \( y_0 \) is of odd
	length, and of the form\\
	$x_0- \cdots -l - m_1 - m_2-r- \cdots -y_0$.
	We root the tree at \( m_2 \).
	By similar reasoning as earlier, player~\( x \) has a safe strategy if
	the edges are not oriented as \( l \leftarrow m_1 \leftarrow m_2 \).
	Then, player~\( x \) can reach \( l \) safely while player~\( y \)
	reaches \( r \) safely, however then player~\( x \) cannot reach \( m_1
	\), therefore has to find a safe position in \( T_l \).
	Therefore, player~\( y \) has a winning strategy if and only if all edges
	in \( T_l \) are downwards once more.
\end{proof}

\section{Girth at least \( 6 \)}\label{sec:g6}
The first extension of our results on trees (acyclic connected graphs) are graphs
with high girth (only big cycles).
We show that with strong connectivity, players are essentially always safe from
one another -- unless one player is at a parent of the other player in the
initial position.
\begin{thm}\label{thm:char6}
	If \( G \) is strongly connected and \( g \geq 6 \), the minmax value of
	a pair of vertices is \( 0 \) if and only if they are not neighbors.
	For neighboring vertices, the player at the child of the other has payoff in
	the range \( \left[ -\frac{4(1-\delta)}{4-\delta}; -(1-\delta) \right]
	\).
	In particular, there is a static equilibrium at any pair of vertices at
	distance~\( 3 \) from each other, a \( 2 \)-chase equilibrium for any
	starting vertices with a (directed) \( 2 \)-path from one to the other,
	and a WT~equilibrium starting at every vertex.
\end{thm}
\begin{proof}
	We show the following property: for any pair of vertices that are
	distance at least \( 2 \) from each other, each player can ensure they
	remain distance at least \( 2 \) from each other next round.
	If the vertices are distance at least \( 3 \) from one another it is clear,
	each player can simply remain where they are.
	If the players are distance \( 2 \) from one another, since \( g \geq 6
	\) there is a unique (undirected) path of length \( 2 \) between them.
	Let this path be \( x_0 - m - y_0 \).
	Since the graph is strongly connected, each vertex has degree at least \( 2
	\): \( x_0 \) has some other neighbor \( x_1 \).
	If player~\( x \) moves to \( x_1 \), and player \( y \) moves to any
	neighbor \( y_1 \), it must be that \( d(x_1, y_1) \geq 2 \) to ensure \(
	g \geq 6 \), since there is already an undirected path of length~\( 4 \)
	between \( x_1 \) and \( y_1 \).
	If player \( y \) remains at \( y_0 \) then they are at distance~\( 3 \)
	from one another the next round.
	Therefore, with probability~\( 1 \) over player~\( y \)'s randomization,
	the players remain at distance at least~\( 2 \) from one another.
	This means the minmax value of both players is always \( 0 \) when they
	are distance at least \( 2 \) from one another.

	Note that this implies \( G \) is edge-decisive for all discount factors
	\( \delta \in \left] 0;1 \right[ \): if \( u \to v \) and a player is at
	\( u \) and the other at \( v \), the player at \( u \) can secure
	strictly positive payoff by remaining at \( u \) until the other player
	moves.
	Since \( g \geq 6 \) they will not move to a parent of \( u \)
	immediately, hence the next step they will either be distance~\( 2 \) or
	\( 0 \) from \( u \).
	In either case, the minmax value is \( 0 \) and therefore the player at
	\( u \) can play the minmax and have strictly positive payoff in total.
	As seen in Remark~\ref{rmk:wtcyc}, this means all WT and \( 2 \)-chase
	equilibria take place on directed cycles.

	The static, \( 2 \)-chase and WT~equilibria exist because the graph is
	strongly connected and any unilateral deviation from these equilibria
	must leave the players distance at least \( 2 \) or exactly~\( 0 \) from
	one another.
	This is clear for the static equilibrium for \( d(x_0, y_0) \geq 3 \) and
	the WT equilibria.
	For the \( 2 \)-chase equilibrium, the chased player is always moving to
	a neighbor that is not on the shortest undirected path between the two
	players, therefore the chaser cannot deviate to a vertex closer than
	distance~\( 2 \) by the same girth argument as before.
	The chased player could choose to remain at their current vertex, in
	which case they would get payoff~\( -1 \) (and the continuation payoff
	would be weakly positive for the chaser, since they can safely remain
	until the deviator is distance at least~\( 2 \) from them, or reaches the
	same vertex as them).
	Finally, the chased player could deviate by moving back towards the
	chasing player, in which case both players would be at the same vertex
	and the continuation payoff would also be \( 0 \) by symmetry.
	Therefore, any profitable deviation must lead the deviating player to be
	distance at least \( 2 \) from the non-deviating player, which cannot be
	profitable as shown above -- there are therefore no profitable
	deviations.

	Therefore the only cases in which the payoff of an equilibrium is
	non-zero is if the initial condition has an edge~\( y_0 \to x_0 \) (up to
	permuting the players).
	First note the value for player \( x \) is at most \( -(1-\delta) \),
	given the strategy for player~\( y \) exposed above in the proof of
	edge-decisive of \( G \).

	We now lower bound the value for player~\( x \).
	If \( x_0 \) has another parent then \( x \) can deterministically move
	to this parent, and \( y \) will not move to \( x_0 \) by edge
	decisiveness.
	The players are now distance at least \( 2 \) from one another and the
	continuation payoff is \( 0 \), hence the value for \( x \) in this
	situation is exactly \( -(1-\delta) \).
	Otherwise, if \( x_0 \) has no other parents, let \( -p \) be the
	continuation payoff of player~\( x \) in this position, such that their
	expected payoff is \( -(1-\delta)(1 + \delta p) \).
	Denote \( T \) the set of parents of \( y_0 \) and \( B \) the set of
	children of \( x_0 \), such that we have \( T \to y_0 \to x_0 \to B \).
	Note that moving to a child of \( y_0 \) is dominated for \( y \).
	Letting \( x \) be the row player, the game matrix is,
	\[
		\bordermatrix{& T & y_0 & x_0 \cr
			y_0 & \geq -1-\delta p & 0 & 1+\delta p\cr
			x_0 & 0 & -1-\delta p & 0 \cr
			B & 0 & 0 & \geq -1-\delta p
		}.
	\]
	Since the value of \( x \) is (weakly) increasing in the coefficients of
	this matrix, we deduce the expected continuation payoff for \( x \) is at
	least \( -1/4 ( 1 + \delta p ) \) by mixing between the three as \(
	\left( \frac{1}{4}, \frac{1}{4}, \frac{1}{2} \right) \).

	Since \( -p \) is also the value of these positions for \( x \), we find
	\( -p \geq -\frac{1}{4} (1 + \delta p) \), therefore \( p \leq
	\frac{1}{4-\delta} \).
	This allows to lower bound the total value by \( -(1 - \delta)(1+\delta
	p) \geq -\frac{4(1-\delta)}{4-\delta} \).
\end{proof}

\begin{cor}
	In a strongly connected graph with \( g \geq 6 \) and initial positions
	$(x_0,y_0)$ with no edge between them, the set of \nashes{} is the set of
	mixed strategies $(\varphi_x,\varphi_y)$ such that at each round, $x$
	puts \( 0 \)~probability on going to a child of $\ball{1}{y_t}$ and
	\textit{vice versa}.
\end{cor}
\begin{proof}
	We first note such distributions always exist: if \( x_0 = y_0 \) then
	any parent of the vertex verifies the condition, and otherwise \( d(x_0,
	y_0) \geq 2 \) and Theorem~\ref{thm:char6} proves the existence of such
	a move.

	If at round $t$ player~$x$ puts nonzero probability on such a vertex $v$,
	then a possible move of player~$y$ is to go to a parent $u$ of $v$
	deterministically.
	This gives $y$ strictly positive expected payoff, as we showed in
	Theorem~\ref{thm:char6} the graph is edge-decisive.
	Therefore, every \nashe{} has the property above.
	Conversely, any pair of distributions that verifies this property is
	clearly a \nashe{} since if a player deviates, they can never reach a
	parent of the other player by definition of their strategy.
\end{proof}

\section{Graphs with no unbalanced small cycles}\label{sec:g4}
For our main result, we combine the results of the two previous section on trees
and graphs with girth at least \( 6 \).
We remove the strong connectivity assumption and replace it with an analysis of the
\emph{block-cut tree} of the graph: biconnected components will be analogous to
the strongly connected graphs of Theorem~\ref{thm:char6} (though they are not
always strongly connected), whereas cut vertices will behave more like tree
vertices seen in Theorem~\ref{thm:tree}.
This results in weakening the assumptions from the previous section in two ways:
strong connectivity is no longer required, and we replace the assumption \( g
\geq 6 \) with the assumption \( g \geq 4 \) and the absence of \emph{small
unbalanced cycles} as subgraphs of the graph.
Let us begin by defining these concepts.
\begin{defi}\label{def:xycut}
	An {\bf\( (x, y) \)-cut vertex} of \( G \) for \( x, y \in V \) is a cut
	vertex of \( G \) such that removing it separates \( x \) and \( y \)
	into two distinct connected components.
\end{defi}

\begin{defi}[\citep{gallai1964elementare,harary1966block}]\label{def:bctree}
	For an undirected graph $G=(V,E)$ define its {\bf block-cut tree} as the
	tree containing a vertex for each maximal biconnected component of \( G \),
	a vertex for each of its cut vertices, and an edge connecting each cut
	vertex to the biconnected components it belongs to.
	We call the \textbf{thinned block-cut tree} of \( G \), denoted \( T(G)
	\), the following transformation of its block-cut tree: for each maximal
	biconnected component of size~\( 2 \) containing two cut vertices, remove
	its vertex from the tree and add an edge between its two cut vertices;
	for all other biconnected components of size~\( 2 \), remove its vertex
	and replace it with a vertex labeled by its non-cut vertex.
	A biconnected component remaining in the thinned block-cut tree
	(equivalently, a biconnected component with strictly more than \( 2 \)
	vertices) is called a \textbf{nontrivial biconnected component}.
\end{defi}
Notice each vertex in the thinned block-cut tree is labeled either by a maximal
nontrivial biconnected component or a vertex of the graph, and all vertices
labeled by a vertex of the graph are either cut vertices or leaves of the
block-cut tree.

\begin{lmm}\label{lmm:xycut}
	Along any shortest undirected path between \( x \) and \( y \), the
	indices containing \( (x, y) \)-cut vertices are always the same, and
	each index always contains the same cut vertex.
	Moreover, in all shortest undirected paths the indices that do not
	contain \( (x, y) \)-cut vertices correspond to a vertex in the common
	biconnected component of the previous and next \( (x, y) \)-cut vertices
	in the path.
\end{lmm}
\begin{proof}
	The first part is true by optimal substructure: all paths between the
	players must go through this \( (x,y) \)-cut vertex, and if it is not at
	index~\( i \) along some shortest path then it means that the path from
	either player to the cut vertex is shorter in this path, leading to an
	overall shorter shortest path.
	For the second part of the proof, first note that two successive \( (x,y)
	\)-cut vertices \( a \) and \( b \) on a path share a biconnected
	component -- if they did not, there would exist a cut vertex that
	separates their biconnected components, which in turn must be a \( (x,y)
	\)-cut vertex and be on the path between \( a \) and \( b \).
	Assume by contradiction there is a vertex on a shortest path that is not
	a \( (x,y) \)-cut vertex and does not belong to the biconnected component
	of its surrounding \( (x,y) \)-cut vertices.
	This means the shortest path reaches a cut vertex~\( u \) (to leave the
	biconnected component), then another vertex outside the component, but to
	return to the next \( (x,y) \)-cut vertex, its only path is to traverse
	the cut vertex~\( u \) again: this is a contradiction.
\end{proof}

\begin{figure}
	\begin{subfigure}{.25\textwidth}
		\centering
		\begin{tikzpicture}[every node/.style={shape=circle, fill=black}]
			\node (a) at (-0.588, -0.809) {};
			\node (b) at (-0.952, 0.309) {};
			\node (c) at (0, 1) {};
			\node (d) at (0.952, 0.309) {};
			\node[fill=white,draw] (e) at (0.588, -0.809) {};
			\graph[edges={thick,>=stealth}]{
				(a) -> (b) -> (c) -> (d) -> (e);
				(a) -> (e);
			};
		\end{tikzpicture}
		\caption{\( C^5_{4,1} \).}
	\end{subfigure}\hfill
	\begin{subfigure}{.25\textwidth}
		\centering
		\begin{tikzpicture}[every node/.style={shape=circle, fill=black}]
			\node (a) at (-0.588, -0.809) {};
			\node (b) at (-0.952, 0.309) {};
			\node (c) at (0, 1) {};
			\node[fill=white,draw]  (d) at (0.952, 0.309) {};
			\node (e) at (0.588, -0.809) {};
			\graph[edges={thick,>=stealth}]{
				(a) -> (b) -> (c) -> (d);
				(a) -> (e) -> (d);
			};
		\end{tikzpicture}
		\caption{\( C^5_{3,2} \).}
	\end{subfigure}\hfill
	\begin{subfigure}{.25\textwidth}
		\centering
		\begin{tikzpicture}[every node/.style={shape=circle, fill=black}]
			\node (a) at (0, 0) {};
			\node (b) at (0, 1.7) {};
			\node[fill=white,draw] (c) at (1.7, 1.7) {};
			\node (d) at (1.7, 0) {};
			\graph[edges={thick,>=stealth}]{
				(a) -> (b) -> (c);
				(a) -> (d) -> (c);
			};
		\end{tikzpicture}
		\caption{\( C^4_{2,2} \).}
	\end{subfigure}\hfill
	\begin{subfigure}{.25\textwidth}
		\centering
		\begin{tikzpicture}[every node/.style={shape=circle, fill=black}]
			\node (a) at (0, 0) {};
			\node (b) at (0, 1.7) {};
			\node (c) at (1.7, 1.7) {};
			\node[fill=white,draw] (d) at (1.7, 0) {};
			\graph[edges={thick,>=stealth}]{
				(a) -> (b) -> (c) -> (d);
				(a) -> (d);
			};
		\end{tikzpicture}
		\caption{\( C^4_{3,1} \).}
	\end{subfigure}
	\caption{The four unbalanced small cycles, with their minimal vertex
	highlighted.}\label{fig:unbal}
\end{figure}
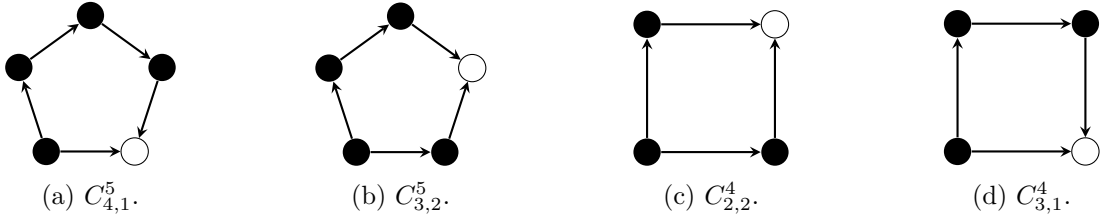

\begin{defi}
	Let \( C^k_{a,b} \) be the length~\( k \) undirected cycle with \( a \)
	consecutive edges in one direction and the remaining \( b \) edges in the
	opposite direction.
	We call the \textbf{unbalanced small cycles} the four cycles \(
	C^5_{4,1}, C^5_{3,2}, C^4_{2,2} \) and \( C^4_{3,1} \), illustrated in
	Figure~\ref{fig:unbal}.
\end{defi}

In Appendix~\ref{app:cycle-directions} we find the remaining 4-cycles and 5-cycles, which are not unbalanced.

We now characterize winning positions in graphs with no small unbalanced cycles.
The proof of the following theorem is deferred to Theorem~\ref{thm:char4app} in
Appendix~\ref{app:g4}, we offer a proof sketch here.
\begin{thm}\label{thm:char4}
	Suppose \( G \) satisfies \( g \geq 4 \) and does not contain any of the
	unbalanced small cycles as subgraphs (in particular, this is verified
	when \( g \geq 6 \)).
	Consider the thinned block-cut tree~\( T(\unG) \) of $G$.
	For a given initial position $(x_0,y_0)$ in $G$, player~\( y \) has strictly
	positive payoff over player~\( x \) if and only if there exists a
	shortest path between \( x \) and \( y \) wich is is either,
	\begin{enumerate}
		\item of even length, of the form $x_0- \cdots -l \leftarrow m
		\leftarrow r- \cdots -y_0$, where $d(x_0,m)=d(y_0,m)$, the midway
		vertex $m$ is an $(x,y)$-cut vertex of $G$ and the subtree~$T_l$
		(containing player $x$) of $T(\unG)$ rooted at \( m \) has no
		nontrivial biconnected components or upwards edges; or
		\item of odd length, of the form $x_0- \cdots -l \leftarrow
		m_1\leftarrow m_2-r- \cdots -y_0$, where $d(x_0,m_1)=d(y_0,m_2)$,
		the vertex $m_1$ is an $(x,y)$-cut vertex of $G$ and the
		subtree~$T_l$ (containing player $x$) of $T(\unG)$ rooted at \(
		m_1 \) has no nontrivial biconnected components or upwards edges.
	\end{enumerate}
\end{thm}
\begin{proof}[Proof sketch]
	First note that by Lemma~\ref{lmm:xycut} the criteria are well-defined,
	i.e.\ they do not depend on the chosen shortest path.
	The main idea is to show that both players have a safe strategy when they
	are both in a nontrivial biconnected component and at distance at least
	\( 2 \) from one another.

	Indeed, suppose player~\( x \) is at distance exactly \( 2 \) from
	player~\( y \).
	For staying at their current node or moving towards \( y \) not to be
	safe strategies, there must be a directed \( 2 \)-path from \( y \) to \(
	x \).
	Moreover, \( x \) must have some other neighbor in the biconnected
	component than the one between \( x \) and \( y \): moving to this
	neighbor not being safe must mean \( y \) is a parent of that neighbor or
	is adjacent to a parent of that neighbor.
	The constructed edges so far create a \( 4 \) or \( 5 \) cycle with some
	orientations fixed: one can check that no orientations for the remaining
	edges avoid creating a small unbalanced cycle.
	The reasoning when \( y \) is outside of the biconnected component is
	similar, because \( y \) can only enter the component through a unique
	cut vertex.

	The rest of the characterization is similar to Theorem~\ref{thm:tree},
	with the added subtlety of the case where the root is part of a
	non-trivial biconnected component.
	In these cases, we show that before reaching the midpoint, both players
	will enter its biconnected component, and be distance at least \( 2 \)
	from one another.
	Since we have shown these are both safe positions, we deduce that for a
	player to have a winning strategy, the root must be a cut vertex.
	From there, we have shown that nontrivial biconnected components are
	safe, therefore the connected component containing the losing player must
	be composed of only cut vertices, which makes it a tree.
	By similar reasoning to Theorem~\ref{thm:tree} once more, we show this
	tree must be outdirected from the midway point.
\end{proof}

We now characterize the presence of cycle-based equilibria under the assumptions
of Theorem~\ref{thm:char4}.
We show the weakest necessary condition one could hope for (under edge
decisiveness, which we show holds for \( G \)) is necessary and sufficient:
cycle-based equilibria exist if and only if a directed cycle is present.
\begin{prop}
	If \( G \) satisfies \( g \geq 4 \) and does not contain any of the
	unbalanced small cycles as subgraphs, \( G \) has a WT~equilibrium and a
	\( 2 \)-chase equilibrium if and only if there is a directed cycle in \(
	G \).
	In particular, \( G \) always has either a cycle-based or a static
	equilibrium.
\end{prop}
\begin{proof}
	As we saw earlier in the proof of Theorem~\ref{thm:char4}, \( G \) is
	edge-decisive.
	In particular, by Remark~\ref{rmk:wtcyc}, this implies that WT and \(
	2 \)-chase equilibria are only supported by directed cycles.

	A directed cycle is clearly necessary to have either a \( 2 \)-chase or a
	WT~equilibrium.
	Consider a directed cycle of~\( G \), it is contained in a nontrivial
	biconnected component.
	As shown in Theorem~\ref{thm:char6}, any profitable deviation from
	cycle-based equilibria leaves the players distance either \( 0 \) or \( 2
	\) from one another.
	Since we have shown in Theorem~\ref{thm:char4} that the value of such
	pairs of vertices for a player in a nontrivial biconnected component is at
	least \( 0 \), there are no profitable deviations.
\end{proof}

We state a necessary and sufficient condition under the previous assumptions for
there to be a static equilibrium, albeit for concision we state it in negative
form.
The proof of the following proposition is deferred to
Proposition~\ref{prop:nostaticapp} in Appendix~\ref{app:g4}.
\begin{prop}\label{prop:nostatic}
	If \( G \) satisfies \( g \geq 4 \) and contains no unbalanced cycles, \(
	G \) has no static equilibria if and only if the following are all true,
	\begin{enumerate}
		\item \( G \) has exactly one nontrivial
		biconnected component~\( B \);
		\item The thinned BC-tree is an outdirected tree
		rooted at \( B \) (all edges go downwards);
		\item \( B \) is of diameter exactly \( 2 \) and all pairs of
		distance-\( 2 \) vertices have a common neighbor that is a parent
		of one of the two;
		\item Every vertex in \( B \) has a parent.
	\end{enumerate}
\end{prop}

We finally show that up to added outgoing branches, the only graph with no static
equilibria with no small unbalanced cycles and \( g \geq 5 \) is the directed \(
5 \)-cycle.
The proof of the following corollary is deferred to Corollary~\ref{cor:wt2capp}
in Appendix~\ref{app:g4}.
\begin{cor}\label{cor:wt2c}
	If \( g \geq 5 \) and small unbalanced cycles are forbidden, the only
	graphs with no static equilibria are composed of a directed \( 5 \)-cycle
	with outgoing edges from its nodes forming a directed outgoing tree
	rooted at the cycle.
	In particular, all graphs with \( g \geq 5 \) and no small unbalanced
	cycles either have a static equilibrium or both a \( 2 \)-chase and a
	WT~equilibrium.
\end{cor}

\section{Constructions with no cycle-based or static equilibria}\label{sec:constructions}
We now argue that unbalanced cycles play an important role in the existence of
cycle-based equilibria by exposing constructions without cycle-based or static
equilibria.
We show that even under strong connectivity (a much stronger assumption than
before, meaning all vertices are part of a nontrivial biconnected component and
of a directed cycle) and the absence of \emph{any} \( 4 \)-cycles, these
equilibria do not always exist.
Exhibiting such counter-examples is subtle since ensuring strong connectivity
often creates many new directed cycles, which creates opportunities for
cycle-based equilibria.

We first show that a constant upper bound on \( \delta \) along with a girth
assumption ensures the edge-decisiveness of a graph, which will be used in the
proofs in this section.
The proof of the following lemma is deferred to Lemma~\ref{lmm:gammaapp} in
Appendix~\ref{app:constructions}.
\begin{lmm}\label{lmm:gamma}
	Let \( \gamma_a \) be the unique positive root of \( \gamma^{a-2}+\gamma-1=0
	\) for \( a \ge 4 \).
	If \( g(G) \geq a \), then \( G \) is edge-decisive for all \( \delta <
	\gamma_a \).
\end{lmm}
In particular, when $g=5$ we obtain that $G$ is edge-decisive for all $\delta<\gamma_5\approx0.68233$.

\begin{figure}
	\begin{subfigure}{.24\textwidth}
		\centering
		\begin{tikzpicture}[every node/.style={shape=circle, fill=black}]
	\node (a)[fill=black,text=white,draw] at (0.5872183921658105, -0.8095675141099725) {$a$};
	\node (b)[fill=black,text=white,draw] at (-0.5885451895120709, -0.8095675141099725) {$b$};
	\node (c)[fill=black,text=white,draw] at (-0.9508714314867199, 0.3095860474607925) {$c$};
	\node (d)[fill=black,text=white,draw] at (0.0009066448532779603, 1.0008916914407422) {$d$};
	\node (e)[fill=black,text=white,draw] at (0.9512915839797025, 0.30865728931841013) {$e$};
	\node (f)[fill=black,text=white,draw] at (0.5872183921658105, 1.8796295738690973) {$f$};
	\node (g)[fill=black,text=white,draw] at (-0.5885451895120709, 1.8796295738690973) {$g$};
	\node (h)[fill=white,draw] at (1.2494008343953413, 2.4223855549864036) {};
	\node (i)[fill=white,draw] at (1.5589868818561337, 3.373256986473124) {};
	\node (j)[fill=white,draw] at (0.6081154503694137, 3.682843033933916) {};
	\node (k)[fill=white,draw] at (0.2985294029086213, 2.7319716024471957) {};
	\graph[edges={thick,>=stealth}]{
		(c) -> (b) -> (a) -> (e);
		(c) -> (d) -> (e);
		(e) -> (f) -> (g) -> (c);
		(f) -> (h) -> (i) -> (j);
		(f) -> (k) -> (j);
	};
\end{tikzpicture}
		\caption{A graph with no static equilibria.}\label{fig:nost}
	\end{subfigure}
	\hfill
	\begin{subfigure}{.37\textwidth}
		\centering
		\begin{tikzpicture}[scale=.3,every node/.style={shape=circle, draw=black,
	font=\small, inner sep=0pt, color=lightgray, thick, minimum size=10pt}]
	  \node[minimum size=25pt,color=black,very thick] (Paris) at (0,0) {\color{black}1};
	  \node[minimum size=18pt,color=black,thick] (Brest) at (-6,4) {\color{black}2};
	  \node[minimum size=18pt,color=black,thick] (Bordeaux) at (-6,-5) {\color{black}3};
	  \node[minimum size=18pt,color=black,thick] (Marseille) at (6,-5) {\color{black}4};
	  \node[minimum size=18pt,color=black,thick] (Lille) at (6,4) {\color{black}5};
	  \node[minimum size=18pt,color=black,thick] (LeHavre) at (0,6) {\color{black}6};
	  \node[minimum size=18pt,color=black,thick] (Montpellier) at (0,-9) {\color{black}7};
	  \node[minimum size=18pt,color=black,thick] (Bastia) at (-10,-10) {\color{black}8};
	  \node (Amiens) at (3,2) {15};
	  \node (Lyon) at (3,-2) {14};
	  \node (Clermont) at (2,-3.8) {134};
	  \node (Tours) at (-2,-2) {13};
	  \node (Nantes) at (-3,2) {12};
	  \node (MiniLeHavre) at (0,3) {16};
	  \node (MiniMontpellier) at (0,-6.5) {71};
	  \node (Toulouse) at (-6,-9) {37};
	  \node (Rodez) at (-3,-5) {34};
	  \node (MiniBastia) at (-3,0) {81};

	  \graph[edges={>=stealth}]{(Brest) -> (LeHavre);
	  (LeHavre) -> (Lille);
	  (Lille) -> (Marseille);
	  (Lille) -> (Amiens);
	  (Amiens) -> (Paris);
	  (Brest) -> (Nantes);
	  (LeHavre) -> (MiniLeHavre);
	  (MiniLeHavre) -> (Paris);
	  (Nantes) -> (Paris);
	  (Tours) -> (Paris);
	  (Clermont) -> (Tours);
	  (Lyon) -> (Paris);
	  (Marseille) -> (Lyon);
	  (Marseille) -> (Rodez);
	  (Rodez) -> (Bordeaux);
	  (Bordeaux) -> (Brest);
	  (Toulouse) -> (Bordeaux);
	  (Montpellier) -> (Toulouse);
	  (Montpellier) -> (Marseille);
	  (Paris) -> (MiniMontpellier);
	  (MiniMontpellier) -> (Montpellier);
	  (Paris) -> (MiniBastia);
	  (Bordeaux) -> (Tours);
	  (Marseille) -> (Clermont);};

	  \draw[->,>=stealth] (MiniBastia) to [out=180, in=55] (Bastia);
	  \draw[->,>=stealth] (Bastia) to [out=90, in=250] (Brest);
	  \draw[->,>=stealth] (Bastia) to [out=0,in=270] (Marseille);
\end{tikzpicture}
		\caption{A graph with no WT~equilibria.}\label{fig:nowt}
	\end{subfigure}
	\hfill
	\begin{subfigure}{.37\textwidth}
		\centering
		\begin{tikzpicture}[scale=.4]
	  \node[n] (1) at (0,12) {1};
	  \node[n] (2) at (3,12) {2};
	  \node[n] (3) at (6,12) {3};
	  \node[n] (4) at (0,9) {4};
	  \node[n] (5) at (3,9) {5};
	  \node[n] (6) at (6,9) {6};
	  \node[n] (7) at (10,10.5) {7};
	  \node[n] (8) at (3,6) {8};
	  \node[n] (9) at (-3,10.5) {9};

	  \draw[diredge] (1) -- (2);
	  \draw[diredge] (2) -- (3);
	  \draw[diredge] (3) -- (4);
	  \draw[diredge] (4) -- (5);
	  \draw[diredge] (5) -- (6);
	  \draw[diredge] (6) -- (7);
	  \draw[diredge] (7) to [out=270,in=50] (8);
	  \draw[diredge] (8) to [out=130,in=270] (9);
	  \draw[diredge] (9) -- (1);
	  \draw[diredge] (6) -- (1);
	  \draw[diredge] (9) -- (4);
	  \draw[diredge] (3) -- (7);
\end{tikzpicture}
		\caption{A graph with no \( 2 \)-chase equilibria.}\label{fig:no2chase}
	\end{subfigure}
	\caption{Constructions without cycle-based or static equilibria.}\label{fig:counter}
\end{figure}
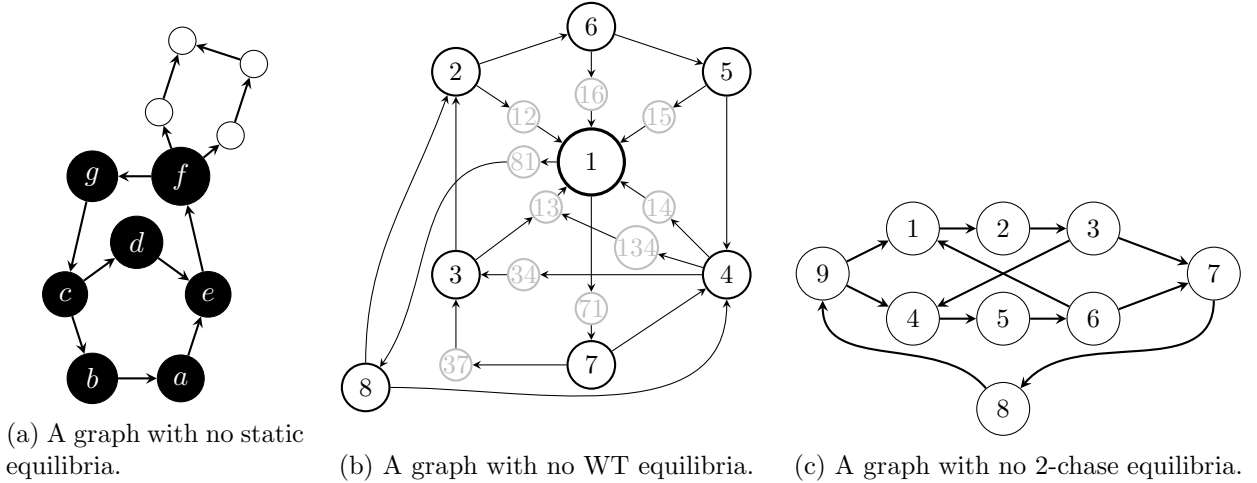

\medskip{}
We begin with a construction (Figure~\ref{fig:nost}) that supports no static
equilibria, whilst only containing one type of small unbalanced cycles (\(
C^5_{3,2} \)) and no cycles of length \( 4 \).
To be relevant, this construction must violate at least one item of
Proposition~\ref{prop:nostatic}: we show it violates all of them but the last,
which cannot be violated by any such construction (a vertex with no parents always
has a static equilibrium).
We note that unlike the other constructions, this construction is not strongly
connected: it can easily be made to be by removing the white vertices while
preserving its properties (except it would no longer violate item~\ref{item:1} of
Proposition~\ref{prop:nostatic}).
The proof of the following theorem is deferred to Theorem~\ref{thm:nostapp} in
Appendix~\ref{app:constructions}.
\begin{thm}\label{thm:nost}
	The graph in Figure~\ref{fig:nost} has $g\ge5$, only contains \(
	C^5_{3,2} \) of the small unbalanced cycles and has no static equilibria
	for every \( \delta \in \left] 0; 1 \right[ \).
	Moreover, it violates every condition of Proposition~\ref{prop:nostatic}
	except for item~\ref{item:5} (which clearly is a necessary condition).
\end{thm}

We now show a construction (Figure~\ref{fig:nowt}) that supports no WT equilibria
for all \( \delta<\gamma_5 \) (in this case the construction is edge-decisive).
This construction is more involved as it contains many directed cycles and
deviations take more rounds to be profitable.
Intuitively, it supports no WT~equilibria because it has a central vertex
(vertex~\( 1 \) in the figure) that has many parents. No matter which parent
the WT~equilibrium prescribes to go to, there is always a parent to which a deviator can profitably deviate.
We then show that every directed cycle that does not go through \( 1 \) verifies
a similar property with vertex~\( 2 \).
The proof of the following theorem is deferred to Theorem~\ref{thm:nowtapp} in
Appendix~\ref{app:constructions}.
\begin{thm}\label{thm:nowt}
	The graph in Figure~\ref{fig:nowt} is strongly connected, satisfies
	\( g\ge5 \), and has no WT~equilibria for every \( \delta<\gamma_5\approx0.68233 \).
\end{thm}

We conclude this section with a construction (Figure~\ref{fig:no2chase}) that
supports no \( 2 \)-chase equilibria.
It is based on two \( C^5_{4,1} \) cycles offering profitable deviations in all
directed cycles, at vertices \( 6 \) and \( 3 \).

\begin{thm}
	The graph in Figure~\ref{fig:no2chase} is strongly connected, satisfies
	$g\ge5$, only contains \( C^5_{4,1} \) of the small unbalanced cycles and
	has no 2-chase equilibria for every $\delta<\gamma_5\approx0.68233$.
\end{thm}
\begin{proof}
	Without loss of generality assume that $y$ is the chaser. If $x$ is at
	vertex~\( 9 \), then $y$ is at 7 and can deviate to either 3 or 6, based on whether $\varphi_x(9,7)(4)=1$ or $\varphi_x(9,7)(1)=1$ (respectively), and gain positive payoff.
	If $x$ is at 7 or 8, the same argument shows \( y \) can profitably
	deviate.
	If $x$ is at 6, then $y$ is at 4. Then, if $\varphi_x(6,4)(7)=1$ then $y$ can deviate to 3, and if $\varphi_x(6,4)(1)=1$ then $y$ can deviate to 9 and gain positive payoff.
	If $x$ is at 3, then $y$ is at 1. If $\varphi_x(3,1)(7)=1$ then $y$ deviates to 6, and if $\varphi_x(3,1)(4)=1$ then $y$ will profit from deviating to 9.
	If $x$ is at 1, 2, 4 or 5, then we go back to the previous arguments.

	Finally, by Lemma~\ref{lmm:gamma} with $a=g(G)=5$, after $y$ gains positive payoff, $y$ can avoid incurring negative payoff for at least 3 rounds, making $y$'s deviation overall profitable for every $\delta<\gamma_5$.
\end{proof}

Observe that the graphs in Figures~\ref{fig:nowt} and~\ref{fig:no2chase} are nonplanar by \textsc{Kuratowski}'s theorem~\cite{kuratowski1930probleme}: Figure~\ref{fig:nowt} contains a subdivision of $K_{3,3}$ (the complete bipartite graph with 3 vertices on each side) with the vertices $1,2,4$ and $3,6,8$ on the two sides of $K_{3,3}$; Figure~\ref{fig:no2chase} contains a subdivision of $K_{3,3}$ with the vertices $1,4,7$ and $3,6,9$ on the two sides of $K_{3,3}$.

\section{Outerplanar graphs}\label{sec:outer}
Say that $G$ is \emph{outerplanar} if it is planar and all of its vertices are
part of the unbounded face of $G$.
In this section we prove that all strongly connected outerplanar graphs with $g\ge4$ have both a WT and a 2-chase equilibria. We note that they do not necessarily have static equilibria, for example the 4-cycle.

In this section, suppose that $G$ is outerplanar and strongly connected.
Fix some outerplanar embedding of $G$.
As a planar graph, $G$ consists of bounded faces $C_1,\ldots,C_k$ (which we also
call \emph{minimal (undirected) cycles}).
Since all the vertices of $G$ touch the unbounded face of $G$, the cycles $C_1,\ldots,C_k$ are connected to each other either by a vertex or by an edge. These cycles are minimal in the sense that each $C_i$ bounds exactly one face.

We now prove that $G$ contains a well-directed minimal cycle. The proof of this property works for any planar graph.

\begin{lmm}\label{lmm:wdface}
	All strongly connected planar graphs have a well-directed face.
\end{lmm}
\begin{proof}
	Since it is strongly connected, $G$ contains a well-directed cycle~$C$.
	Fix some planar embedding of \( G \).
	Let $H_C$ be the subgraph of $G$ containing \( C \) and all vertices and
	edges inside of $C$ in the planar embedding.
	Then $H_C$ is also planar and strongly connected, by definition.
	If $H_C$ is a bounded face, the proof is done.
	Otherwise, there is a vertex $v\in V(H_C)\setminus V(C)$ that also neighbors some $u\in V(C)$.
	Suppose without loss of generality that $u\to v\in E(H_C)$.
	By strong connectivity, there is a directed path from $v$ to $u$ inside $H_C$ (possibly containing edges of $C$). By taking a shortest such path, we can assume that it does not contain $u\to v$.
	By then concatenating this path with the edge $u\to v$, we construct a
	well-directed cycle $C'$ in $H_C$.
	Let $H_{C'}$ be the subgraph corresponding to $C'$ and edges and vertices
	contained in \( C' \) in the embedding.
	Then we have $H_{C'}\subset H_C$, because $H_{C'}\subseteq H_C$ and
	there is an edge $(u\to v)\in E(C)\setminus E(H_{C'})$.
	Similarly to $H_C$, we also have that $H_{C'}$ is planar and strongly connected. Now we apply the whole argument above to $H_{C'}$ instead of $H_C$.
	Since $\card{E(H_{C'})} < \card{E(H_C)}$, and $G$ is finite, this process
	will eventually terminate.
	Upon termination we reach a minimal cycle, concluding the proof.
\end{proof}

\begin{thm}\label{thm:outerplanar}
	If $G$ is outerplanar then it supports both a WT and a 2-chase
	equilibria.
\end{thm}
\begin{proof}
	By Lemma~\ref{lmm:wdface}, let $C$ be a well-directed and minimal cycle
	in $G$.
	We show that if either player deviated from any position in a \( 2
	\)-chase or WT~equilibrium on \( C \), they either fail to secure
	strictly positive payoff, or end up in another position in either the \( 2
	\)-chase equilibrium or the WT~equilibrium.
	This shows no deviation is strictly profitable, since a profitable
	deviation would need to leave this set of positions.

	\paragraph{2-chase equilibrium}
	We begin with a \( 2 \)-chase equilibrium on \( C \), such that $y$ is
	chasing \( x \).

	Suppose $y$ deviates from a \( 2 \)-chase equilibrium. Let $N$ be the vertex from
	which $y$ deviates, let $B$ be the vertex $y$ deviates to.
	Let $P\to N\to N_1\to N_2\to A\to A_1\subseteq E(C)$ be a segment of $C$ on which the players walk in a 2-chase profile
	(so possibly $[A=P\wedge A_1=N]\vee A_1=P$ holds).
	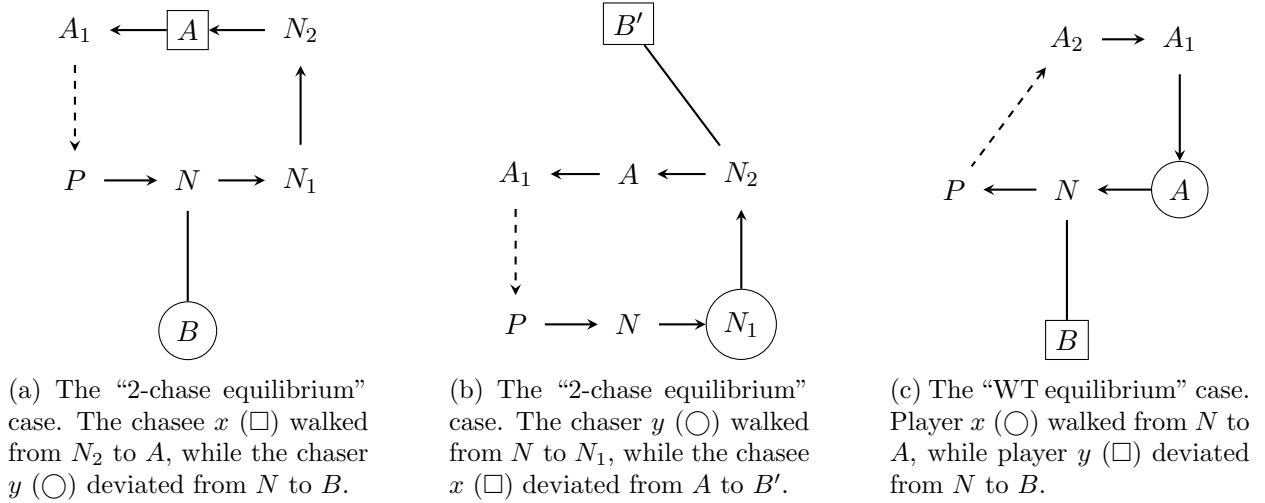
\begin{figure}
		\begin{subfigure}{.29\textwidth}
			\centering
			\begin{tikzpicture}[every node/.style={shape=circle, minimum size=0.1cm}]
				\node (P) at (0,0) {$P$};
				\node (N) at (1.5,0) {$N$};
				\node (N1) at (3,0) {$N_1$};
				\node (N2) at (3,2) {$N_2$};
				\node[shape=rectangle,draw] (A) at (1.5,2) {$A$};
				\node (A1) at (0,2) {$A_1$};
				\node[draw] (B) at (1.5,-2) {$B$};
				\graph[edges={thick,>=stealth}]{
					(P) -> (N) -> (N1) -> (N2) -> (A) -> (A1) ->[dashed] (P);
					(N) -- (B);
				};
			\end{tikzpicture}
			\caption{The ``2-chase equilibrium'' case. The chasee $x$ ($\square$) walked from $N_2$ to $A$, while the chaser $y$ ($\bigcirc$) deviated from $N$ to $B$.}
		\end{subfigure}\hfill
		\begin{subfigure}{.29\textwidth}
			\centering
			\begin{tikzpicture}[every node/.style={shape=circle, minimum size=0.1cm}]
				\node (P) at (0,-2) {$P$};
				\node (N) at (1.5,-2) {$N$};
				\node[draw] (N1) at (3,-2) {$N_1$};
				\node (N2) at (3,0) {$N_2$};
				\node (A) at (1.5,0) {$A$};
				\node (A1) at (0,0) {$A_1$};
				\node[shape=rectangle, draw] (B) at (1.5,2) {$B'$};
				\graph[edges={thick,>=stealth}]{
					(P) -> (N) -> (N1) -> (N2) -> (A) -> (A1) ->[dashed] (P);
					(N2) -- (B);
				};
			\end{tikzpicture}
			\caption{The ``2-chase equilibrium'' case. The chaser $y$ ($\bigcirc$) walked from $N$ to $N_1$, while the chasee $x$ ($\square$) deviated from $A$ to $B'$.}
		\end{subfigure}\hfill
		\begin{subfigure}{.29\textwidth}
			\centering
			\begin{tikzpicture}[every node/.style={shape=circle, minimum size=0.1cm}]
				\node (P) at (0,0) {$P$};
				\node (N) at (1.5,0) {$N$};
				\node[draw] (A) at (3,0) {$A$};
				\node (A1) at (3,2) {$A_1$};
				\node (A2) at (1.5,2) {$A_2$};
				\node[shape=rectangle, draw] (B) at (1.5,-2) {$B$};
				\graph[edges={thick,>=stealth}]{
					(P) <- (N) <- (A) <- (A1) <- (A2) <-[dashed] (P);
					(N) -- (B);
				};
			\end{tikzpicture}
			\caption{The ``WT equilibrium'' case. Player $x$ ($\bigcirc$) walked from $N$ to $A$, while player $y$ ($\square$) deviated from $N$ to $B$.}
		\end{subfigure}
		\caption{Illustrations for Theorem~\ref{thm:outerplanar}.}
	\end{figure}

	So $x$ is at $A$ and $y$ is at $B$. Suppose that $x$ counters the deviation by staying in place. Let $A'$ be an ancestor of $A$ to which $y$ gets to gain positive payoff (there could be multiple ancestors, as $y$ can use a mixed strategy). To preserve outerplanarity, if $y$ reenters the cycle $C$ then it does so at $P$, $N$ or $N_1$. There are two cases for $A'$.

	If $A'=N_2$, then $y$ gets to $N_2$ via $N_1$ or $A$.
	First, suppose it is the $N_1$ case. Then $x$ will be at $A$ and $y$ will be at $N_1$, so it is as if $y$ never deviated. Since the payoff in a 2-chase equilibrium is 0, this makes $y$'s deviation not winning.
	Now suppose it is the $A$ case. $x$ is also at vertex~$A$, so $x$ can now start playing according to a WT strategy on the cycle $C$. From here we can apply the argument in the ``WT equilibrium'' case, because at least one round has passed since $y$'s deviation.

	If $A'\ne N_2$, then $A'$ belongs to some other cycle $C'$ that is adjacent to $C$. Since $G$ is outerplanar, $C$ and $C'$ share the vertex $A$, and possibly exactly one of the vertex-edge pairs $A_1,(A\to A_1)$ or $N_2,(N_2\to A)$. Therefore, $y$ will reach $A'$ by exiting $C$ from either $N_2$, $A$ or $A_1$. Observe that the case of $A$ falls into the other two cases, and that the case of $N_2$ falls into the argument of the previous paragraph. So the case of $A_1$ is left.
	Here, upon observing $y$ exiting from $C$ to $C'$ by moving away from
	$A_1$, player $x$ can move from $A$ to $N_2$, and stay there.
	While $x$ stays at $N_2$, player $y$ has to reenter $C$ via $A_1$
	or $A$. Since $g\ge4$, only one of $d(y,A)=1,d(y,A_1)=1$ can hold at any
	given round (particularly when $y\notin V(C)$). So the counter-strategy
	of $x$ from this point is as follows: if $y$ is at vertex~$A'$, move to
	$N_1$. If $y$ is at vertex~$A$, move to $N_2$ and continue according to the chaser strategy of the 2-chase profile on $C$. If $y$ is in $C'$ and not at $A'$, move to $N_2$. If $y$ is at $N$, go to $A$ and continue according to the original chasee strategy.
	To see that $x$ does not incur negative payoff, observe that as long as $y$ is not back in $C$, the two players are at distance at least 2 from each other. When $y$ is back in $C$, it can approach $x$ from either $A$ or $N_1$. In both cases, $x$ starts playing according to a 2-chase profile on $C$. Since this profile yields 0 payoff for both players, this shows that its deviation is not profitable.

	Suppose $x$ deviates from the \( 2 \)-chase strategy profile on \( C \).
	Suppose $x$ deviates from $N_2$ to $B'$, and $y$ goes from $N$ to $N_1$.
	Then, again by outerplanarity, $x$ can reach an ancestor of $y$ only by
	reaching $N$, $N_1$ or $N_2$.
	In the case of $N$, observe that by outerplanarity, $x$ has to go through $P$ to reach $N$ (or $N$, but we can just focus on the first time $x$ visits $N$ after its deviation). Since $g\ge4$, we have $N_2\to N\notin E(G)$. Furthermore, since $x$ deviates from a chasee strategy of a 2-chase profile on $C$, it will reach $P$ at least one round after deviating. Therefore, $y$ can observe $x$ reaching $P$ after its deviation, and then start playing the chasee strategy of a 2-chase profile on $C$. Since the players get 0 payoff from a 2-chase profile, the deviation of $x$ is not profitable.
	In the case of $N_1$, note that we reach a WT state, hence the deviation is not profitable once again.
	The case of $N_2$ is left. Observe that here, $y$ can play the same counter-strategy that $x$ plays when $y$ deviates and exits $C$ via $A_1$ (in the previous paragraph), making $x$'s deviation here unprofitable.
	Overall, $x$ and $y$ have no profitable deviations, showing that the 2-chase
	profile on $C$ is a \nashe{}.

	\paragraph{WT equilibrium}
	Define the WT equilibrium to walk on $C$ from each vertex to its parent.

	Suppose without loss of generality that player $y$ deviates from a WT
	walk. Let $N$ be the last vertex before deviation.
	Let $P\to N\to A\to A_1\to A_2$ be a segment of $C$ on which the players walk in a WT profile.
	Suppose that $x$
	counters the deviation by staying at $A$ until $y$ is at a neighbor of an
	ancestor of $A$. Suppose that $y$ reaches an ancestor $A'$ of $A$, and
	suppose that $y$ moves to $A'$ from a vertex $A_3$ ($y$ can use a mixed strategy during deviation).

	Suppose first that $A'=A_1$. Then since $g\ge4$, player $y$ can reach $A_1$ from either $P$ or $A$ after deviating from $N$. If $y$ reaches $A$, then it is as if $y$ has never deviated, so $x$ keeps playing accordingly, and $y$'s deviation is unprofitable because the payoff from a WT profile is 0. Otherwise, $y$ goes to $P$. Then, $x$ can start playing the chaser strategy in the 2-chase profile on $C$. By the ``2-chase equilibrium'' case, this shows that $y$'s deviation is not winning.

	Now suppose that $A'\ne A_1$. Here we follow the same reasoning as in the ``2-chase equilibrium'' case when ``$A\ne N_2$'' (using the notation from that case). Then $A'\notin V(C)$.
	Again, to reach $A'$, player $y$ needs to exit $C$ through $N$ to a cycle
	$C'\ne C$ that is adjacent to $C$, or through one of $A$ or $A_1$. The cases of $A$ and $A_1$ were already proved, so the case of $N$ is left.
	There, observe that $x$ can play the same counter-strategy that it plays in the ``2-chase equilibrium'' case, when $y$ deviates and exits $C$ via $A_1$ (in the notation of that case), making $t$'s deviation unprofitable here too.
	Overall, $y$ does not have a winning deviation, and this is a \nashe{}, concluding the proof.
\end{proof}

In light of this result and given the constructions of
Section~\ref{sec:constructions}, we conjecture the following,
\begingroup\setlength{\topsep}{4pt}
\begin{conj}
	All strongly connected planar graphs with \( g \geq 4 \) have either a
	static, a \( 2 \)-chase or a WT~equilibrium.
\end{conj}
\endgroup

\section{Computing equilibria}\label{sec:comp}
For every pair of vertices, the value of minmax play when the players are situated
at these vertices must verify a system of \textsc{Bellman} equations:
\begin{equation}\label{eq:bellman}
	V_x(u, v) = \max_{s_x \in \Delta \left( \ball{1}{u} \right)} \; \min_{s_y
	\in \Delta \left( \ball{1}{v} \right)} \left( \E_{\left( u', v' \right)
	\sim s_x \times s_y} \left[ (1-\delta) r_x(u', v') + \delta V_x(u', v')
	\right] \right).
\end{equation}
This can be computed using value iteration~\citep{shapley1953stochastic}, thus
leading to efficient computation of optimal player strategies.
Once \( V_x \) is computed for all pairs of vertices, all mixed equilibria can be
deduced by computing the set of all minmaxes satisfying each state's equation in
equation~\eqref{eq:bellman}.

We now show that the same can be said of pure equilibria, using
Algorithm~\ref{alg:pure}.
Recall the \textbf{strong product graph} \( \unG \boxtimes \unG \) is defined as
the graph with vertices \( V^2 \) and edges \( E' \) such that \( ((u,v),(u',v'))
\in E' \) if and only if \( u=v \) and \( (u',v') \in \unG \) or \( (u,v) \in
\unG \) and \( u'=v' \) or \( (u,v) \in \unG \) and \( (u',v') \in \unG \).

\begin{algorithm}
\caption{Algorithm computing pure equilibria.}\label{alg:pure}
\begin{algorithmic}
	\Require{The strong product graph \( \unG \boxtimes \unG \).}
	\Ensure{A subgraph~\( \mathcal{F}_\infty \) of \( \unG \boxtimes \unG \)
	indicating possible states in a pure equilibrium.}
	\Ensure{A labeling~\( \ell: V^2 \to V \) for each state of possible
	moves.}
	\State{\( t \gets 0 \)}
	\State{\( \ell(u, v) \gets \ball{1}{u} \) for \( (u,v) \in V^2 \)}
	\State{\( \mathcal{F}_0 \gets \revE \)}
	\While{\( \mathcal{F}_t \neq \mathcal{F}_{t-1} \)}
	\State{\( \mathcal{F}_{t+1} \gets \mathcal{F}_t \)}
	\For{\( (u,v) \in V^2 \setminus \mathcal{F}_{t+1} \)}
	\For{\( u' \in \ell(u,v) \)}
	\If{\( \exists v' \in \ball{1}{v} \mid (u', v') \in \mathcal{F}_t \)}
	\State{\( \ell(u,v) \gets \ell(u,v) \setminus \{ u' \} \)}
	\EndIf{}
	\EndFor{}
	\If{\( \ell(u,v) = \emptyset \)}
	\State{\( \mathcal{F}_{t+1} \gets \mathcal{F}_{t+1} \cup \{ (u,v) \} \)}
	\EndIf{}
	\EndFor{}
	\State{\( t \gets t+1 \)}
	\EndWhile{}
	\State{\Return{\( \mathcal{F}_\infty = \mathcal{F}_{t-1} \cup \{ (v,u)
	\mid (u,v) \in \mathcal{F}_{t-1} \} \)}}
\end{algorithmic}
\end{algorithm}

\begin{prop}\label{prop:algpure}
	Pure \nashes{} all have payoff~\( 0 \) at every round after the initial
	condition, and can be found using Algorithm~\ref{alg:pure}.
\end{prop}
\begin{proof}
	We show that if one player uses a deterministic strategy, the other can
	indefinitely avoid negative payoff.
	Indeed, if the opponent is moving to a parent of the player's current
	vertex, they can also move to that parent and ensure \( 0 \)~payoff.
	If they are not moving to a parent of the player's current vertex, they can
	remain at their current vertex and ensure \( 0 \)~payoff as well.
	Therefore each player's best responses ensure weakly positive payoff,
	which means any pure equilibrium is \( 0 \)-payoff.

	Next, we show Algorithm~\ref{alg:pure} is correct.
	For this, we show that \( \mathcal{F}_t \) contains states in which the
	second player can win in \( t \) moves.
	This is clear for \( t = 0 \), and the induction step follows from the
	fact that \( (u,v) \) leads to a win of the second player within \( t+1
	\) moves if and only if for each possible move of the first player, the
	second player has a response that leads to a win within \( t \) moves.
	Since a winning move by the second player (when it exists) is always
	better than a non-winning move, this means that any state in \(
	\mathcal{F}_t \) for finite \( t \) cannot be part of an equilibrium
	play.
	By symmetric reasoning, no state in \( \mathcal{F}_\infty \) can be in
	any equilibrium play.
	Note this also implies the \textsf{While} loop terminates: any winning
	position for the second player is winning in at most \( \card{V^2} < +
	\infty \) moves.

	Finally, the states which are not in \( \mathcal{F}_\infty \) are part of
	equilibrium play since each player has a move such that all moves of the
	other player lead to weakly positive payoff for them.
	Hence, any profile which traverses these states constitutes a pure
	equilibrium.
\end{proof}

\subsection*{Acknowledgments}
The authors would like to thank Laurel~Britt and Stephen~Morris for valuable
insights.

\printbibliography{}

\newpage
\appendix
\section{Cycle directions}\label{app:cycle-directions}
Denote $D_5=\langle \sigma,\tau\mid\sigma^5=\tau^2=\tau\sigma\tau\sigma=e \rangle$ the dihedral group with $2\cdot5=10$ elements. Fix a labeling of the vertices of the
5-cycle and let $\mathcal{O}_5$ be the set of orientations of the 5-cycle. Observe
that $\card{\mathcal{O}_5}=2^5=32$. Let $D_5$ act on $\mathcal{O}_5$ in the natural way.
By the \textsc{Cauchy}-\textsc{Frobenius} theorem, the number of orbits in $\mathcal{O}_5$ with respect to $D_5$'s action, namely the number of distinct orientations, up to the permutations in $D_5$ that do not affect our results,
\begin{align*}
	\card{\mathcal{O}_5/D_5}=\frac{1}{10}\sum_{i\in\s{0,1}}\sum_{j\in\s{0,1,\ldots,4}}\card{\mathcal{O}_5^{\tau^i\sigma^j}}
\end{align*}
Observe that only the two well-directed orientations are fixed with respect to $\sigma^j$. Moreover, $\mathcal{O}_5^{\tau\sigma^j}=\emptyset$ for every $j$. Therefore,
\begin{align*}
	\card{\mathcal{O}_5/D_5}=\frac{32+4\cdot2}{10}=4
\end{align*}

Similarly, we consider $D_4$ acting on $\mathcal{O}_4$. $\card{\mathcal{O}_4^{\tau^i\sigma^j}}$ is 16 for $i=j=0$, 2 for $j\in\s{1,3}$, and 4 for $j=2$. Therefore,
\begin{align*}
	\card{\mathcal{O}_4/D_4}=\frac{16+4\cdot2+2\cdot4}{8}=4
\end{align*}

Therefore, there are two 4-cycles and two 5-cycles which are not unbalanced, which are described in Figure~\ref{fig:bal}.

\begin{figure}
	\begin{subfigure}{.25\textwidth}
		\centering
		\begin{tikzpicture}[every node/.style={shape=circle, fill=black}]
			\node (a) at (-0.588, -0.809) {};
			\node[fill=white,draw] (b) at (-0.952, 0.309) {};
			\node (c) at (0, 1) {};
			\node[fill=white,draw] (d) at (0.952, 0.309) {};
			\node (e) at (0.588, -0.809) {};
			\graph[edges={thick,>=stealth}]{
				(a) -> (b) <- (c) -> (d) <- (e) -> (a)
			};
		\end{tikzpicture}
		\caption{\( C^5_{4,1} \).}
	\end{subfigure}\hfill
	\begin{subfigure}{.25\textwidth}
		\centering
		\begin{tikzpicture}[every node/.style={shape=circle, fill=black}]
			\node (a) at (-0.588, -0.809) {};
			\node (b) at (-0.952, 0.309) {};
			\node (c) at (0, 1) {};
			\node (d) at (0.952, 0.309) {};
			\node (e) at (0.588, -0.809) {};
			\graph[edges={thick,>=stealth}]{
				(a) -> (b) -> (c) -> (d) -> (e) -> (a)
			};
		\end{tikzpicture}
		\caption{\( C^5_{3,2} \).}
	\end{subfigure}\hfill
	\begin{subfigure}{.25\textwidth}
		\centering
		\begin{tikzpicture}[every node/.style={shape=circle, fill=black}]
			\node (a) at (0, 0) {};
			\node (b) at (0, 1.7) {};
			\node (c) at (1.7, 1.7) {};
			\node (d) at (1.7, 0) {};
			\graph[edges={thick,>=stealth}]{
				(a) -> (b) -> (c) -> (d) -> (a)
			};
		\end{tikzpicture}
		\caption{\( C^4_{2,2} \).}
	\end{subfigure}\hfill
	\begin{subfigure}{.25\textwidth}
		\centering
		\begin{tikzpicture}[every node/.style={shape=circle, fill=black}]
			\node (a) at (0, 0) {};
			\node[fill=white,draw] (b) at (0, 1.7) {};
			\node (c) at (1.7, 1.7) {};
			\node[fill=white,draw] (d) at (1.7, 0) {};
			\graph[edges={thick,>=stealth}]{
				(a) -> (b) <- (c) -> (d) <- (a)
			};
		\end{tikzpicture}
		\caption{\( C^4_{3,1} \).}
	\end{subfigure}
	\caption{The \( 4 \) and \( 5 \)-cycles which are not
	unbalanced.}\label{fig:bal}
\end{figure}
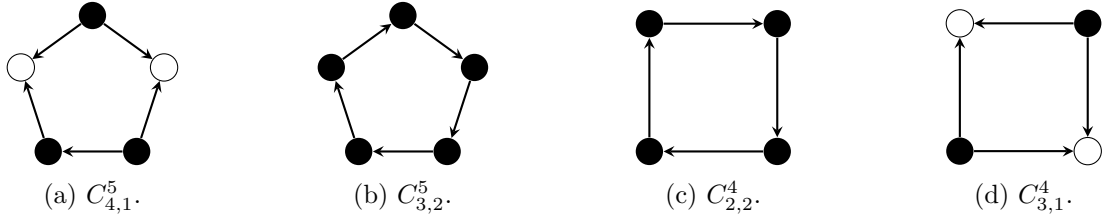

\section{Proofs for graphs with no unbalanced small cycles}\label{app:g4}
\begin{thm}\label{thm:char4app}
	Suppose \( G \) satisfies \( g \geq 4 \) and does not contain any of the
	unbalanced small cycles as subgraphs (in particular, this is verified
	when \( g \geq 6 \)).
	Consider the thinned block-cut tree~\( T(\unG) \) of $G$.
	For a given initial position $(x_0,y_0)$ in $G$, player~\( y \) has strictly
	positive payoff over player~\( x \) if and only if there exists a
	shortest path between \( x \) and \( y \),
	\begin{enumerate}
		\item of even length, of the form $x_0- \cdots -l \leftarrow m
		\leftarrow r- \cdots -y_0$, where $d(x_0,m)=d(y_0,m)$, the midway
		vertex $m$ is an $(x,y)$-cut vertex of $G$ and the subtree~$T_l$
		(containing player $x$) of $T(\unG)$ rooted at \( m \) has no
		nontrivial biconnected components or upwards edges; or
		\item of odd length, of the form $x_0- \cdots -l \leftarrow
		m_1\leftarrow m_2-r- \cdots -y_0$, where $d(x_0,m_1)=d(y_0,m_2)$,
		the vertex $m_1$ is an $(x,y)$-cut vertex of $G$ and the
		subtree~$T_l$ (containing player $x$) of $T(\unG)$ rooted at \(
		m_1 \) has no nontrivial biconnected components or upwards edges.
	\end{enumerate}
\end{thm}
\begin{proof}
	We first show the criteria are well-defined, i.e.\ they do not depend on
	the chosen shortest path.
	By Lemma~\ref{lmm:xycut}, the midway vertex is either always a \( (x, y)
	\)-cut vertex or never, and the endpoints of the midway edge each are
	either always a \( (x, y) \)-cut vertex or never.
	This also means that the second condition cannot be verified for \( x_0
	\) over \( y_0 \) and \( y_0 \) over \( x_0 \) simultaneously: if it
	were, then both endpoints of the midway edge would be \( (x, y) \)-cut
	vertices, meaning the unique midpoint edge between them would be in both
	directions.

	We then show that both players have a safe strategy when they are both
	in a nontrivial biconnected component and at distance at least \( 2 \)
	from one another.
	For this, we show player~\( x \) always can play a move that ensures
	player~\( y \) is either distance at least \( 2 \) or at a child of \( x
	\) at the end of each round.
	If player~\( y \) is distance~\( 3 \) or more, player~\( x \) can remain
	where they are.
	If player~\( y \) is distance exactly~\( 2 \), they must have a
	neighbor that is a parent of player~\( x \) (otherwise player~\( x \) can
	remain).
	Player~\( x \) could move to this parent, unless player~\( y \) has a
	directed edge to it as well (because \( g \geq 4 \)).
	We have so far established a directed two-path from player~\( y \) to
	player~\( x \).
	Consider another neighboring vertex to player~\( x \) within its
	biconnected component (it has at least \( 2 \)): if this neighbor is not
	a safe move it must be a child of player~\( y \), or a child of a
	neighboring vertex to player~\( y \).
	In the former case, we have exposed a \( 4 \)-cycle with the directions
	of three of its edges contradicting each other.
	In the latter case, we have exposed a \( 5 \)-cycle, along with
	directions for \( 3 \) of its edges.
	One can easily check that there are no possible directions for the
	remaining two edges without introducing one of the two forbidden
	subgraphs \( C^5_{4,1} \) or \( C^5_{3,2} \).

	Now if player~\( y \) is outside of the biconnected component, the same
	reasoning applies except player~\( x \) can pretend that player~\( y \)
	is at the unique cut vertex in the biconnected component separating it
	from player~\( y \).
	There are three cases: if \( x \) is on the cut vertex~\( u \) separating
	them from player~\( y \), then since \( y \) is distance at least \( 2
	\), player~\( x \) can move to an arbitrary neighbor of the cut vertex
	within the biconnected component and still remain distance~\( 2 \) from
	\( y \).
	The next case is if \( x \) is at a neighbor of \( u \), once more
	player~\( y \) is distance at least~\( 2 \) away and therefore in the
	next move player~\( x \) can move to another neighbor in the biconnected
	component that is distance~\( 2 \) from \( u \) (and thus from \( y \)).
	Finally, if player~\( x \) is distance at least \( 2 \) from \( u \),
	they can remain where they are since we have shown that by the time that
	\( y \) reaches \( u \), \( x \) will have a safe strategy.

	\bigskip{}
	We now show that $G$ is edge-decisive. Consider an edge \( u \to v \), we split two cases depending on whether
	\( u \) is in a nontrivial biconnected component or not.
	If \( u \) is in such a component, then the player at \( u \) can wait
	for the other player to leave \( v \).
	Since \( g \geq 4 \), the other player will be distance \( 2 \) or \( 0
	\), and by the argument above the minmax value will be \( 0 \).
	If \( u \) is not in a nontrivial biconnected component, then it is a cut
	vertex and it has \( 1 \) or \( 2 \) neighbors which are both cut
	vertices.
	The only path the other player has to \( u \) or any of its parents is
	through \( v \), therefore the player can stay at \( u \) until the other
	player reaches \( u \), after which the payoff will be \( 0 \).

	\bigskip{}
	We now show the characterization: consider a shortest undirected path
	between \( x_0 \) and \( y_0 \) in~\( G \), find its midpoint~\( m \) or
	its middle edge~\( m_1 \to m_2 \).

	\paragraph{First case: midpoint vertex}
	Assume the shortest path is of even length, let \( m \) be a midway
	vertex.
	If \( m \) is a \( (x, y) \)-cut vertex~\( c_m \), root~\( T \) at \( c_m
	\), otherwise root it at \( b_M \) where \( M \) is the nontrivial
	biconnected component of the \( (x,y) \)-cut vertices surrounding \( m \)
	(as shown in Lemma~\ref{lmm:xycut}, \( b_M \) is uniquely defined even if
	\( m \) is not).

	If the root is a \( (x, y) \)-cut vertex~\( c_m \), the same proof as
	Theorem~\ref{thm:tree} shows that player~\( y \) has a winning strategy
	over player~\( x \) if and only if \( T_l \) is a tree of cut vertices
	with all of its edges going downwards (since we have shown nontrivial
	biconnected components are safe).

	If the root is a biconnected component~\( b_M \), both players have a
	safe strategy.
	Indeed, when player~\( x \) reaches \( b_M \), they do so at the \( (x,
	y) \)-cut vertex on their side of \( m \) on the path, and they do so
	strictly before the midpoint (because the midpoint is not a \( (x, y)
	\)-cut vertex), hence player~\( y \) is at distance at least~\( 2 \) from
	them.
	As shown above, being in a nontrivial biconnected components while the
	other player is at distance at least~\( 2 \) ensures \( 0 \)~payoff.

	\paragraph{Second case: midpoint edge}
	If the players are at an odd distance from one another, let \( m_1 \) and
	\( m_2 \) be the endpoints of an edge at the midpoint between the two
	players on a shortest path, where \( m_1 \) is on the side of player~\( x
	\) of the path.
	If \( m_1 \) is not a \( (x, y) \)-cut vertex, then player~\( x \) can
	ensure payoff~\( 0 \) by reaching the \( (x,y) \)-cut vertex of the
	biconnected component of \( m_1 \) strictly before \( m_1 \), at which
	point player~\( y \) is distance at least~\( 3 \) from them.
	Therefore, for a player to have negative payoff their midpoint edge end
	must be a \( (x, y) \)-cut vertex: assume \( m_1 \) is a \( (x, y) \)-cut
	vertex, and root the tree at \( m_1 \).

	If player~\( y \) has no shortest path to \( m_1 \) that goes through an
	ancestor of \( m_1 \), then player~\( x \) can safely reach \( m_1 \),
	and at which point player~\( y \) will be distance at least~\( 2 \) from
	them (to avoid negative payoff).
	If player~\( y \) has some longer path leading to an ancestor of \( m_1
	\), this means that \( m_1 \) is part of a nontrivial biconnected
	component and therefore player~\( x \) is safe.
	Otherwise, all edges on a path from \( x \) to \( y \) now are outgoing
	from \( m_1 \), and it is safe for player~\( x \) to remain at \( m_1 \).

	Conversely, if player~\( y \) has a shortest path to \( m_1 \) that goes
	through an ancestor of \( m_1 \), call it \( m_2 \to m_1 \), then
	player~\( x \) cannot safely reach \( m_1 \) (and staying at \( m_2 \) is
	safe for player~\( y \) since \( m_1 \) is an \( (x,y) \)-cut vertex).
	Player~\( x \) is therefore safe if and only if they can find a
	nontrivial biconnected component or an upwards edge in their maximal
	subtree~\( T_l \) (including the edges connecting them to \( m_1 \)).
\end{proof}

\begin{prop}\label{prop:nostaticapp}
	If \( G \) satisfies \( g \geq 4 \) and contains no unbalanced cycles, \(
	G \) has no static equilibria if and only if the following are all true,
	\begin{enumerate}
		\item\label{item:1} \( G \) has exactly one nontrivial
		biconnected component~\( B \);
		\item\label{item:2} The thinned BC-tree is an outdirected tree
		rooted at \( B \) (all edges go downwards);
		\item\label{item:3} \( B \) is of diameter exactly \( 2 \);
		\item\label{item:4} In \( B \), all pairs of distance-\( 2 \)
		vertices have a common neighbor that is a parent of one of the two;
		\item\label{item:5} Every vertex in \( B \) has a parent.
	\end{enumerate}
\end{prop}
\begin{proof}
	Assume the conditions are all true: there is clearly no static
	equilibrium where both vertices are in distinct biconnected components from
	item~\ref{item:1}, and none where both players are in the same
	biconnected component by items~\ref{item:3} and~\ref{item:4}.
	Moreover, there is no maximal vertex in the graph both players can stay at
	together by points~\ref{item:2} and~\ref{item:5}.
	The two cases left are when both vertices are outside of \( B \) or when one
	is inside and one outside.
	If one is outside of \( B \) and one inside, then by item~\ref{item:3}
	and item~\ref{item:5}, the one inside can deviate to a neighbor of a
	parent of the cut vertex separating the player outside from \( B \).
	This means the player outside can no longer safely reach the cut vertex,
	and therefore the player inside can walk down towards the player outside
	and secure positive payoff as in Lemma~\ref{lmm:3path}.
	In both cases the deviation is profitable by item~\ref{item:2} and the
	proof of Lemma~\ref{lmm:3path}.
	If both players are outside \( B \), if they are different distances from
	\( B \) then the closest one can reach \( B \) at least two rounds before
	the other player if the former deviates (they are at least one vertex closer and the other player does not move the round they deviate).
	After reaching \( B \), the reasoning from the previous case holds: the
	deviating player can move to a neighbor of a parent of the cut vertex of
	the other player, and gain positive payoff by Lemma~\ref{lmm:3path}
	again.
	If both players are the same distance from \( B \), then either of them
	can deviate to reach \( B \) one round before the other player.
	If their cut vertices are neighbors, the player whose cut vertex is
	parent to the other player's cut vertex can profitably deviate (they
	reach their cut vertex one round before the other player).
	Otherwise, by item~\ref{item:4}, one player's cut vertex has a neighbor
	that is a parent of the other player's cut vertex.
	By assuming the former player deviates, they can reach the parent of the
	other player's cut vertex at the same time as the other player can reach
	their cut vertex.

	Conversely, suppose each of the above points is not verified, we expose a
	static equilibrium in each case.
	\begin{enumerate}
		\item If \( G \) has two or more nontrivial biconnected
		components, a static equilibrium is to stay at distance at least
		\( 3 \) from each other, in distinct biconnected components (this
		is always possible because nontrivial biconnected components are
		always diameter at least \( 2 \)).
		If \( G \) has no nontrivial biconnected components, it is a tree
		and has a topologically maximal vertex (for the edge direction
		ordering) at which both players can remain in a static
		equilibrium.
		\item If there is an edge in the thinned BC-tree that is not
		oriented away from \( B \), a static equilibrium is for one
		player to be anywhere in \( B \) (at distance at least \( 2 \))
		and the other at the origin of this upwards edge.
		\item If \( B \) is of diameter~\( 3 \) or more, a static
		equilibrium is for both players to be at two vertices of distance at
		least~\( 3 \) from one another in \( B \).
		This is true because a player in a nontrivial biconnected
		component always has a safe strategy when the other player is
		distance at least \( 2 \) by Theorem~\ref{thm:char4}.
		\item If point~\ref{item:4} is false, then a static equilibrium
		exists at the two violating vertices in \( B \): a deviation
		cannot use any path of length \( 2 \) between them (by
		edge-decisiveness), therefore after deviating the players are
		distance at least \( 2 \) from one another.
		\item If a vertex in \( B \) does not have a parent, it is
		topologically maximal in \( G \) and a static equilibrium is for
		both players to stay at this vertex.\qedhere
	\end{enumerate}
\end{proof}

\begin{cor}\label{cor:wt2capp}
	If \( g \geq 5 \) and small unbalanced cycles are forbidden, the only
	graphs with no static equilibria are composed of a directed \( 5 \)-cycle
	with outgoing edges from its nodes forming a directed outgoing tree
	rooted at the cycle.
	In particular, all graphs with \( g \geq 5 \) and no small unbalanced
	cycles either have a static equilibrium or both a \( 2 \)-chase and a
	WT~equilibrium.
\end{cor}
\begin{proof}
	We know the graph has exactly one nontrivial biconnected component \( B
	\) by Proposition~\ref{prop:nostatic} above.
	Start with two vertices in \( B \) that are distance~\( 2 \) from one
	another.
	By biconnectivity, there exists another shortest path between them, of
	length at most \( 3 \) (by upper bound on the diameter of \( B \)) and at
	least \( 3 \) (by lower bound on the girth of \( G \)).
	We have therefore found an undirected \( 5 \)-cycle: fix the direction of
	one edge \( a \to b \) without loss of generality.
	Then, letting \( c \) be the other neighbor of \( b \) in the undirected
	cycle, item~\ref{item:4} above implies that \( b \to c \) must be
	directed this way.
	By repeating this reasoning we find the \( 5 \)-cycle to be directed.

	Assume by contradiction there is some other vertex~\( u \) in \( B \) which is outside of the
	5-cycle, without loss of generality a neighbor of a vertex of
	the cycle.
	By item~\ref{item:4} again, \( u \) must be a child of its neighbor in
	the cycle.
	By item~\ref{item:3} and \( g \geq 5 \) there must be a \( 2 \)-path from
	\( u \) to one of the two furthest vertices in the \( 5 \)-cycle from \(
	u \).
	Finally, reasoning with item~\ref{item:4} we find the edges in this path
	continue in the same direction, exposing a parent of a vertex in the \( 5
	\)-cycle: this is a contradiction as we have shown above neighbors of the
	cycle must be children of the cycle's vertex.

	Item~\ref{item:2} in Proposition~\ref{prop:nostatic} shows that the
	thinned BC-tree indeed is a directed outgoing tree rooted at \( B \).
	Finally, there is a walking together and a \( 2 \)-chase equilibrium in
	the directed \( 5 \)-cycle.
\end{proof}

\section{Proofs for constructions with no cycle-based or static
equilibria}\label{app:constructions}
\begin{lmm}\label{lmm:gammaapp}
	Let \( \gamma_a \) be the unique positive root of \( \gamma^{a-2}+\gamma-1=0
	\) for \( a \ge 4 \).
	If \( g(G) \geq a \), then \( G \) is edge-decisive for all \( \delta <
	\gamma_a \).
\end{lmm}
\begin{proof}
	Suppose player $x$ is at $u$, player $y$ is at $v$, and $(u\to v)\in E$.
	We show that $x$ can avoid incurring negative payoff for at least $a-2$
	rounds. Let the strategy of $x$ from that point be: if $y$ is not at the
	same vertex as $x$, then stay in place. Otherwise, if $x$ and $y$ are at
	the same vertex, play the minmax strategy (which yields payoff~\( 0 \) by
	symmetry).

	The only possibility for \( y \) to have positive expected continuation
	payoff is by reaching a parent of \( u \) without going through \( u \).
	Since $x$ starts playing the strategy above when $d(x,y)=1$, player $y$ will need at least $g-2$ rounds to gain positive payoff, in particular at least $a-2$ rounds.
	Therefore, the payoff of player $x$ will be at least (up to the
	$(1-\delta)$ factor),
	\begin{align*}
		1-\sum_{k=a-2}^\infty\delta^k=1-\frac{\delta^{a-2}}{1-\delta}
	\end{align*}
	which is strictly positive for all $\delta<\gamma_a$, by definition.
\end{proof}

\begin{thm}\label{thm:nostapp}
	The graph in Figure~\ref{fig:nost} is strongly connected, has $g\ge5$,
	only contains \( C^5_{3,2} \) of the small unbalanced cycles and has no
	static equilibria for every \( \delta \in \left] 0; 1 \right[ \).
	Moreover, it violates every condition of Proposition~\ref{prop:nostatic}
	except for item~\ref{item:5} (which clearly is a necessary condition).
\end{thm}
\begin{proof}
	We begin by showing some useful facts: the only pairs of dark vertices of
	distance at least \( 2 \) in which the first has a positive value other
	the second are \( (c,a), (d,a) \) and \( (g,b) \).
	Indeed, for every other pair of dark vertices, we notice that each player
	can reach the well-directed cycle \( c\to d\to e\to f\to g\to c \)
	safely, after which they can always move counterclockwise in this cycle
	to avoid the other player when they are distance \( 2 \) from them.

	For \( (c,a) \), suppose player \( x \) is at \( a \) and player \( y \)
	at \( c \).
	Player~\( y \) can move to a parent of every vertex
	in \( \ball{1}{a} \): by moving to the parent of whichever one \( x \) is
	likeliest to move to, they ensure strictly positive payoff as long as \(
	d \to e \), \( c \to b \) and \( b \to a \) are decisive edges.
	The former two clearly are as the top of the edge is in the well-directed
	cycle, and \( b \to a \) is because the player at \( b \) can safely move
	to \( c \) next round and be in the well-directed cycle.
	If \( x \) selects one of the two less likely moves, then the
	continuation payoff is at least \( 0 \) for \( y \) as they will be
	distance~\( 2 \) or \( 0 \) from one another in a pair of vertices that is
	not dominating for player~\( x \).

	The positive value of \( (d,a) \) follows from this: the
	player at \( d \) can move to \( c \) if the player at \( a \) is
	likelier to move to \( \{b,a\} \), and otherwise can remain at \( d \) if
	the player at \( a \) is likelier to move to \( e \).
	This works as long as the edges \( d \to e \) and \( c \to b \) are
	decisive, which is the case for the same reasons as above.
	In the case the randomness does not lead to the right outcome for the
	player at \( d \), the continuation payoff is once more (weakly) positive
	for the same reasons.

	The positive value of \( (g,b) \) follows from a similar reason: the
	player at \( g \) can move to \( c \) if the player at \( b \) is
	likelier to move to \( \{b,a\} \), and otherwise can remain at \( g \) if
	the player at \( b \) is likelier to move to \( c \).
	This works as long as the edges \( g \to c \) and \( c \to b \) are
	decisive, which is the case for the same reasons as above.
	In the case the randomness does not lead to the right outcome for the
	player at \( g \), the continuation payoff is once more (weakly) positive
	for the same reasons.

	We now show edge decisiveness of the graph for all \( \delta \in \left]
	0; 1 \right[ \) : edges that have their top in the safe cycle have
	already been shown to be decisive, as well as \( b \to a \).
	Edges that involve a white vertex are also decisive as the player at the
	top can safely reach \( f \), which is part of the safe directed cycle.
	The last remaining edge is \( a \to e \): suppose \( y \) is at \( b \)
	and \( x \) at \( a \).
	Note \( x \) has a safe strategy since they are in the safe directed
	cycle (they can move to \( d \) then stay distance at least \( 2 \) from
	\( y \) in the directed cycle).
	If \( x \) puts higher probability on \( f \) than \( d \) then \( y \)
	can move to \( e \) and obtain positive expected payoff.
	Therefore \( x \) puts higher probability on \( d \) than \( f \) (and
	potentially some probability on \( e \)).
	\( y \) can therefore move to \( b \): if \( x \) moves to \( d \) or \(
	e \) then \( y \) can reach \( c \) safely next round.
	In the event that \( x \) has moved to \( f \), \( y \) can safely move
	to \( a \) then either the initial situation repeats or they can move to
	\( e \) the next round.

	\bigskip{}
	We can now show that every pair of vertices does not support a static
	equilibrium.
	There are no static equilibria where both players are at the same vertex,
	since every vertex has a parent and we have shown edge decisiveness.

	For pairs of vertices distance~\( 2 \) from one another, they must not
	share a neighbor that is a parent to one of the two by edge decisiveness
	of this graph.
	The only such pair of dark vertices is \( (d,a) \): the player at \( d \)
	can however profitably deviate to \( c \), so this is not a static
	equilibrium either.
	There is such a pair of white vertices, but one player can deviate to \( f
	\) in one round after which they have positive expected payoff for the
	same reason as in the case of \( (c,a) \).

	For distance-\( 3 \) pairs of dark vertices, there are \( (g,a) \) and \(
	(f,b) \).
	The former is not a static equilibrium because the player at \( g \) can
	deviate to \( c \) and the latter neither because the player at \( f \)
	can deviate to \( g \).

	Finally, if one player is at a white vertex and another at a dark vertex, the
	player at a dark vertex can always reach \( f \) or a neighbor of \( e \)
	in one step.
	After this, the player in white can no longer reach \( f \) without
	having negative payoff (if they have not already incurred it), and the
	white cycle \( C^5_{3,2} \) reduces to a directed path, which allows us
	to conclude by Lemma~\ref{lmm:3path}.
\end{proof}

\begin{thm}\label{thm:nowtapp}
	The graph in Figure~\ref{fig:nowt} is strongly connected, satisfies
	\( g\ge5 \), and has no WT~equilibria for every \( \delta<\gamma_5\approx0.68233 \).
\end{thm}
\begin{proof}
	We show that the vertices labeled with 1 and 2 cannot be in a cycle induced by a WT-equilibrium walk, by showing that $y$ has a winning deviation.
	This is sufficient by Remark~\ref{rmk:wtcyc} and because all
	well-directed cycles in the graph contain either vertex~\( 1 \) or
	vertex~\( 2 \).

	\paragraph{Vertex~\( 1 \)}
	Suppose $x$ and $y$ are at vertex~\( 1 \). If $\varphi_x(1,1)(12)=1$
	(i.e., the distribution $x$ plays according to when $x$ and $y$ are at
	$1$ puts probability \( 1 \) on moving to \( 12 \)), $y$
	has the following winning deviation: first go to $13$.
	If $\varphi_x(12,13)(1)>\frac{\delta^7}{1+\delta^7-\delta}$, then \( y
	\) stays at $13$ for one round. Otherwise, \( y \) goes to \( 3 \).
	Then, if $x$ is at vertex~\( 2 \), $y$ gains positive payoff.
	If $x$ is at vertex~\( 12 \),
	then depending on whether $x$ is likeliest to go from 12 to 1, 12 or 2;
	go to 13, 2, or 3 respectively.
	If \( x \) is at vertex~1, then either avoid incurring negative payoff for 7 rounds, or play a safe strategy for the rest of the game, as described below.

	We show that this deviation has positive payoff for \( y \).
	We consider the states from which $y$ gains positive payoff within one
	round with positive probability but does not, and show that $y$ can avoid
	incurring negative payoff for long enough. Let $p=\varphi_x(12,13)(1)$.
	We split along two cases depending on the value of $p$. When $x$ is at 12 and
	$p>\frac{\delta^7}{1+\delta^7-\delta}$,
	player~$y$ does not gain positive payoff in the following round if and only if
	$x$ stayed at 12 went or to 2 (and $y$ stayed at 13). If \( x \) stayed at 12, nothing has changed. Otherwise, \( x \) went to 2. Then, observe that $y$ can either avoid incurring negative payoff for at least 7 rounds, or play a safe strategy for the rest of the game, as follows:
	\begin{enumerate}
		\item Go from 13 to 3.
		\item If $x$ is at 12, stay at 3 until $x$ is either at 3, 7 or at one of 14, 5, 8.
		\begin{enumerate}
			\item If $x$ is at 3, then play a minmax strategy (which guarantees a payoff of 0 by symmetry).
			\item If $x$ is at 7, go to 2, and stay there until $x$ is at 37 or 4.
			\begin{enumerate}
				\item If $x$ is at 37, walk on the path $2-6-5$.
				\item If $x$ is at 4, walk to 6. If $x$ is then at 8, walk to 5, otherwise stay.
			\end{enumerate}
			\item\label{item:x-at-14-5-8} If $x$ is at either 14, 5 or 8, walk to 37, and stay there until $x$ is at 71 or 4.
			\begin{enumerate}
				\item If $x$ is at 71, go on the path $37-3-2$.
				\item If $x$ is at 4, go to 7 and stay there until $x$ is at 1 or 7. If $x$ is at 1, go to 37. If $x$ is at 7, play a minmax strategy (which guarantees a payoff of 0 by symmetry).
			\end{enumerate}
		\end{enumerate}

		\item If $x$ is at 2, then because $y$ is at 3 and gains positive
		payoff, so in this case is handled by the edge-decisiveness of $G$ (Lemma~\ref{lmm:gamma}).
		\item If $x$ is at 3, then play a minmax strategy (which guarantees a payoff of 0 by symmetry).
		\item If $x$ is at 8, follow step~\ref{item:x-at-14-5-8}, with the following extension. If $x$ reached 71, then $y$ walks on the path $37-3-2-6-5$, with stopping on vertices,
	\end{enumerate}

	Observe that the strategies in cases 3 and 4 are safe for the rest of the game. Combining the analyses of steps 1, 2 and 5, $y$'s expected payoff is at least (up to a $\delta^t(1-\delta)$ factor for some $t\in\N$),
	\[
		p-(1-p)\sum_{k=7}^\infty\delta^k=p-(1-p)\frac{\delta^7}{1-\delta}>\frac{\delta^7}{1+\delta^7-\delta}-\frac{\delta^7}{1+\delta^7-\delta}=0.
	\]

	In the other case of $p$, we have $p\le\frac{\delta^7}{1+\delta^7-\delta}$. Then, in the last step of $y$'s deviation (as described above), if $y$ does not gain positive payoff in the last round of that deviation description, then observe that $y$ can either avoid incurring negative payoff for at least 7 rounds or play a safe strategy for the rest of the game (one can verify this by considering each of the initial positions $(1,3),(12,3),(1,2),(2,2),(12,13),(2,13)$, where $(u,v)$ means that $x$ is at $u$ and $y$ is at $v$).
	Therefore, $y$'s expected payoff is at least (up to a $\delta^t(1-\delta)$ factor for some $t\in\N$),
	\[
		\frac{1}{3}-\frac{2}{3}\cdot\frac{\delta^7}{1-\delta},
	\]
	which is strictly positive for every $\delta<\rho$.

	By Lemma~\ref{lmm:gamma} with $a=g(G)=5$, after $y$ gains positive
	payoff, because $G$ is edge-decisive, $y$ will have overall positive payoff, making the deviation profitable.

	A similar argument works if $\varphi_x(1,1)(v)=1$ for $v\in\s{14,15,16}$: to see this, construct a similar deviation of $y$ (depending on the distribution $\varphi_x(1,1)$), then one can verify that if $y$ attempts to gain positive payoff in a certain round but fails, then $y$ can avoid incurring negative payoff for at least 7 rounds, or play a safe strategy for the rest of the game.

	Now suppose that $\varphi_x(1,1)(13)=1$. Then $y$ has the following winning
	deviation: go from vertex~\( 1 \) to \( 14 \). Then, if
	$\varphi_x(13,14)(1)>\frac{\delta^7}{1+\delta^7-\delta}$ then stay at
	14 for one round. Otherwise, go to 4. Then, if $x$ is at vertex~\( 134
	\), then $y$ gains positive payoff. If $x$ is at vertex~\( 3 \), then depending on whether $x$ is likeliest to go from 3 to 13, 2, 3, 34 or 37, go to 134, 8, 34, 4, 7 respectively. If $x$ is at 1, then avoid incurring negative payoff for 7 rounds or play a safe strategy for the rest of the game, similarly to the $\varphi_x(1,1)(12)=1$ case above.

	A similar argument to the case of $\varphi_x(1,1)(12)=1$ shows that if $y$ gains positive payoff at a
	given round with positive probability but does not, then $y$ can avoid incurring negative payoff for at least 5 rounds. In total, $y$'s expected payoff here is at least (up to a $\delta^t(1-\delta)$ factor for some $t\in\N$),
	\[
		\frac{1}{5}-\frac{4}{5}\cdot\frac{\delta^7}{1-\delta}>0\iff\delta<\rho
	\]
	by definition.

	\paragraph{Vertex~\( 2 \)}
	Suppose that $x$ and $y$ are at vertex~\( 2 \). If $\varphi_x(2,2)(3)=1$, then there
	are two cases. If $\varphi_x(3,3)(37)=1$, then $y$ can keep walking with $x$,
	and by the definitions of WT~equilibria and edge-decisiveness, they will reach
	vertex~\( 1 \), from which we go back to the previous case.
	If $\varphi_x(3,3)(34)=1$, then $y$ has a winning deviation which is the
	analogous deviation to the winning deviation from the case of vertex~\( 1
	\) treated above (here, the deviation of $y$ begins with going to 37).
	Observe that here too, if $y$ fails to gain positive payoff in a certain round, then $y$ can avoid incurring negative payoff for at least 7 rounds, or play a safe strategy for the rest of the game. Therefore, the same constraint $\delta<\rho$ is sufficient here too. Overall, $y$ has a
	winning deviation from vertex~\( 2 \), when $\varphi_x(2,2)(3)=1$.

	Now suppose that $\varphi_x(2,2)(8)=1$. Then $y$ can keep walking with $x$ and
	they will reach vertex~\( 1 \) (by definition of $\varphi_x$), from which
	we go back to the first case.

	Again, Lemma~\ref{lmm:gamma} implies that when $y$ gains positive payoff, $y$'s deviation will remain profitable (when $y$ plays an optimal deviation).

	Therefore, vertices~\( 1 \) and \( 2 \) cannot be part of any
	WT~equilibrium, concluding the proof.
\end{proof}
\end{document}